\documentclass{article}
\usepackage{ijcai24}

\usepackage{times}
\usepackage{soul}
\usepackage{url}
\usepackage[hidelinks]{hyperref}
\usepackage[utf8]{inputenc}
\usepackage[small]{caption}
\usepackage{graphicx}
\usepackage{amsthm}
\usepackage{booktabs}
\usepackage[switch]{lineno}

\usepackage{amsmath}

\usepackage{tikz}
\usetikzlibrary{shapes}
\usepackage{varwidth}
\usepackage{wrapfig}
\usepackage{array,multirow,makecell}
\usepackage[switch]{lineno}

\usepackage[ruled,vlined,linesnumbered,noend]{algorithm2e}
\newcommand{\nosemic}{\renewcommand{\@endalgocfline}{\relax}}
\newcommand{\dosemic}{\renewcommand{\@endalgocfline}{\algocf@endline}}
\let\oldnl\nl
\newcommand{\nonl}{\renewcommand{\nl}{\let\nl\oldnl}}

\usepackage{todonotes}

\usepackage{pgfplots}
\usepackage{caption}
\usepackage{subcaption}
\usepackage{xcolor,colortbl}
\definecolor{LightCyan}{rgb}{0.88,1,1}
\definecolor{LightCyan2}{rgb}{0.48,0.69,0.99}
\definecolor{LightRed}{rgb}{0.99,0.44,0.43}
\definecolor{blueStrong}{rgb}{0.1,0.16,0.31}
\definecolor{linecolor}{rgb}{0.07,0.33,0.66}
\definecolor{bleu}{rgb}{0.07,0.33,0.76}
\definecolor{bleuvictoria}{rgb}{0.03,0.35,0.61}
\definecolor{bleudepths}{rgb}{0.15,0.39,0.32}
\definecolor{bleumalibu}{rgb}{0.12,0.25,0.96}
\definecolor{rouge}{rgb}{0.83,0.08,0}
\definecolor{racingred}{rgb}{0.74,0.09,0.17}
\definecolor{vert}{rgb}{0.12,0.74,0.13}
\definecolor{gold}{rgb}{1,0.84,0}
\definecolor{gold2}{rgb}{0.99,0.94,0.7}
\definecolor{aurometalsaurus}{rgb}{0.43, 0.5, 0.5}
\definecolor{cadetgrey}{rgb}{0.57, 0.64, 0.69}
\definecolor{lemon}{rgb}{0.94, 0.992, 0.37}
\definecolor{coffee}{rgb}{0.435, 0.305, 0.215}
\definecolor{liver}{rgb}{0.325, 0.294, 0.309}
\definecolor{rust}{rgb}{0.717, 0.155, 0.055}
\definecolor{darkkhaki}{rgb}{0.717, 0.155, 0.055}
\definecolor{tawny}{rgb}{0.74,0.71,0.42}
\definecolor{darkkhaki}{rgb}{0.817, 0.34, 0}
\definecolor{cocoabrown}{rgb}{0.208, 0.157, 0.118}
\definecolor{lilas}{rgb}{0.71, 0.4, 0.82}

\newtheorem{thm}{Theorem}

\newtheorem{lm}{Lemma}

\title{Faster Optimal Coalition Structure Generation via Offline Coalition Selection and Graph-Based Search}

\author{
Redha Taguelmimt$^1$
\and
Samir Aknine$^1$\and
Djamila Boukredera$^2$\and
Narayan Changder$^3$\And
Tuomas Sandholm$^{4,5,6,7}$
\affiliations
$^1$Univ Lyon, UCBL, CNRS, INSA Lyon, Centrale Lyon, Univ Lyon 2, LIRIS, UMR5205, Lyon, France\\
$^2$Laboratory of Applied Mathematics, Faculty of Exact Sciences, University of Bejaia, Bejaia, Algeria\\
$^3$TCG Centres for Research and Education in Science and Technology, Kolkata, India\\
$^4$Computer Science Department, 
  Carnegie Mellon University, Pittsburgh, USA\\
$^5$Strategy Robot, Inc.\\
$^6$Strategic Machine, Inc.\\
$^7$Optimized Markets, Inc.
\emails
redha.taguelmimt@gmail.com, samir.aknine@univ-lyon1.fr, djamila.boukredera@univ-bejaia.dz, narayan.changder@tcgcrest.org, sandholm@cs.cmu.edu
}

\begin{document}

\maketitle

\begin{abstract}
Coalition formation is a key capability in multi-agent systems. An important problem in coalition formation is \textit{coalition structure generation}: partitioning agents into coalitions to optimize the social welfare. This is a challenging problem that has been the subject of active research for the past three decades. 
In this paper, we present a novel algorithm, SMART, for the problem based on a hybridization of three innovative techniques. 
Two of these techniques are based on dynamic programming, where we show a powerful connection between the coalitions selected for evaluation and the performance of the algorithms. 
These algorithms use offline phases to optimize the choice of coalitions to evaluate. The third one uses branch-and-bound and integer partition graph search to explore the solution space. 
Our techniques bring a new way of approaching the problem and a new level of precision to the field. 
In experiments over several common value distributions, we show that the hybridization of these techniques in SMART is faster than the fastest prior algorithms (ODP-IP, BOSS) in generating optimal solutions across all the value distributions.
\end{abstract}

\section{Introduction}

One of the main challenges in coalition formation is the \textit{coalition structure generation (CSG)} problem: partitioning the agents into disjoint exhaustive coalitions so as to maximize social welfare. 
(A \textit{coalition structure} is a partitioning of agents into coalitions.) This is a central problem in artificial intelligence and game theory that captures a number of important applications such as collaboration among trucking companies~\cite{sandholm1997coalitions}, distributed sensor networks~\cite{dang2006overlapping}, etc.

Many algorithms have been developed for this problem. Dynamic programming algorithms~\cite{yeh1986dynamic,rahwan2008improved,michalak2016hybrid,changder2019effective,10097921} find an optimal solution if it is computationally feasible to run them to completion. Anytime algorithms~\cite{sandholm1999coalition,dang2004generating,rahwan2009anytime,Ueda_Iwasaki_Yokoo_Silaghi_Hirayama_Matsui_2010,10098066} provide intermediate solutions during the execution and allow premature termination. Heuristic algorithms~\cite{sen2000searching,Ueda_Iwasaki_Yokoo_Silaghi_Hirayama_Matsui_2010,10.5555/3463952.3464039,9643288} focus on speed and do not guarantee  that an optimal solution is found. 

Even though those algorithms perform well in practice in some cases, hybrid algorithms~\cite{michalak2016hybrid,changder2020odss,Changder_Aknine_Ramchurn_Dutta_2021,ijcai2023p35,10.5555/3635637.3663204} that combine dynamic programming with integer partition graph search have emerged as the dominant approach to find optimal solutions to this problem. The fastest exact algorithms to date are hybrid solutions called ODP-IP~\cite{michalak2016hybrid}, ODSS~\cite{changder2020odss}, and BOSS~\cite{Changder_Aknine_Ramchurn_Dutta_2021} that combine IDP~\cite{rahwan2008improved} and IP~\cite{rahwan2009anytime}. IDP is based on dynamic programming and computes the optimal solution for $n$ agents by computing an optimal partition of all the coalitions $\mathcal{C}$ of size $|\mathcal{C}| \in \{2,...,\frac{2n}{3},n\}$. In contrast, IP uses an integer representation of the search space and computes the optimal solution by traversing in a depth-first manner multiple search trees and uses branch-and-bound to speed up the search. 
However, the worst-case run time of the state-of-the-art hybrid algorithms is determined by their respective dynamic programming parts, which still need improvement. 
Also, the hybridization of IDP and IP in these algorithms relies heavily on the effectiveness of IP. Thus, the time required by the algorithms grows considerably when IP is not fast enough. 
Moreover, these algorithms exhibit very high run times for some  distributions.

In light of this, and to enable faster generation of optimal coalition structures, we develop a new algorithm that combines three complementary techniques to guide the search. The advantage of these techniques is threefold. The first technique,  \textit{Complementarity-Based Dynamic Programming (CDP)}, enables SMART to have the best worst-case time performance of all algorithms to date. \textit{GRadual seArch with Dynamic Programming (GRAD)} enables it to find the optimal solution quickly by exploring a minimum number 
of solution subspaces which shortens the run time. \textit{Distributed Integer Partition Search (DIPS)} further accelerates the search by exploring the subspaces that are most likely to contain the optimal solution. In short, our main contributions are:
\begin{itemize}
    \item We develop a novel algorithm for optimal CSG that combines three new techniques, 
    resulting in a significant performance improvement. Two of these techniques use offline phases to optimize the search. Moreover, we introduce a new complementarity principle in dynamic programming, where the optimal solution is found by combining the evaluation results of two distinct sets of coalitions. We also propose another principle of gradual search in dynamic programming, where percentages of solution subspaces are searched separately. These principles bring a new way of approaching the problem and a new level of precision to the field. 
    \item We devise the fastest dynamic programming algorithm to date, which bounds the run time, and we propose a new way to speed the search for optimal solutions, while exploring only a part of the search space.
    \item We show that our algorithm outperforms existing algorithms when generating optimal solutions. We show that it is i) orders of magnitude faster in producing optimal solutions, and ii) more stable in the run time when varying the distributions and the numbers of agents.
\end{itemize}

\section{Preliminaries}

The input to a CSG problem is a set of agents $\mathcal{A}$ and a characteristic function $v$.  We say that a CSG problem $\mathcal{A} = \{a_{1}, a_{2}, . . . , a_{n}\}$ is of size $n$. A coalition $\mathcal{C}$ in $\mathcal{A}$ is any non-empty subset of $\mathcal{A}$. The size of $\mathcal{C}$ is $|\mathcal{C}|$, which is the number of agents it contains. A size set is a set of coalition sizes. In a CSG problem, a characteristic function $v$ assigns a real value to each coalition $\mathcal{C}$. A coalition structure $\mathcal{CS}$ is a partition of the set of agents $\mathcal{A}$ into disjoint coalitions. 
Given a set of non-empty coalitions $\{\mathcal{C}_{1}, \mathcal{C}_{2},...,\mathcal{C}_{k}\}$, $\mathcal{CS} =\{\mathcal{C}_{1}, \mathcal{C}_{2},...,\mathcal{C}_{k}\}$, where $k= |\mathcal{CS}|$, $\bigcup_{j=1}^{k} \mathcal{C}_{i} = \mathcal{A}$ and for all $i,j \in \{1,2,...,k\}$ where $i \neq j$, $\mathcal{C}_{i} \cap \mathcal{C}_{j} = \emptyset$. 
$\Pi(\mathcal{A})$ denotes the set of all coalition structures. The value of a coalition structure $\mathcal{CS}$ is $V (\mathcal{CS}) = \sum_{\mathcal{C} \in \mathcal{CS}} v(\mathcal{C})$. 
The optimal solution of the CSG problem is the most valuable coalition structure $\mathcal{CS}^{*} \in \Pi(\mathcal{A})$, that is, $\mathcal{CS}^{*}= \mbox{argmax}_{\mathcal{CS} \in \Pi(\mathcal{A})} V(\mathcal{CS})$.

The \textit{integer partition graph}~\cite{rahwan2007near} (see Figure \ref{IPG}) divides the search space into subspaces that are represented by integer partitions of $n$.  Given $n$ agents, each integer partition of $n$ is represented by a node, where the nodes are divided into levels. Each level $l \in \{1,2,..,n\}$ contains nodes representing integer partitions of $n$ that contain $l$ parts. For instance, level 3 contains nodes where integer partitions of $n$ have 3 parts. Two adjacent nodes are connected if the integer partition in level $l$ can be reached from the one in level $l-1$ by splitting only an integer. Each integer partition $\mathcal{P}$ represents a \textit{set} of coalition structures in which the sizes of the coalitions match the parts of $\mathcal{P}$. For example, the node [1,1,2] represents all coalition structures that contain two coalitions of size 1 and one coalition of size 2. Figure \ref{IPG} shows a four-agent example of the integer partition graph. 

For the remainder of this paper, we use the terms solution subspace and node interchangeably.

\begin{figure}[h!]
\centering
\resizebox{1.0\columnwidth}{4.5cm}{%
\begin{tikzpicture}
\large
\node[draw,rectangle,rounded corners=3pt] (a)at(7.5,0){$[1,1,1,1]$};

\node[draw,rectangle,rounded corners=3pt] (a)at(7.5,-2){$[1,1,2]$};

\node[draw,rectangle,rounded corners=3pt] (a)at(5.5,-4){$[1,3]$};

\node[draw,rectangle,rounded corners=3pt] (a)at(9.5,-4){$[2,2]$};

\node[draw,rectangle,rounded corners=3pt] (a)at(7.5,-6){$[4]$};

\node[] (a)at(-0.5,-6){$L_{1}: $};
\draw[>=latex,black!70,dashed,very thick] (7.0,-6) -- (-0.0,-6);
\node[] (a)at(-0.5,-4){$L_{2}: $};
\draw[>=latex,black!70,dashed,very thick] (0.1,-4) -- (-0.0,-4);
\node[] (a)at(-0.5,-2){$L_{3}: $};
\draw[>=latex,black!70,dashed,very thick] (6.5,-2) -- (-0.0,-2);
\node[] (a)at(-0.5,0){$L_{4}: $};
\draw[>=latex,black!70,dashed,very thick] (6.25,0) -- (-0.0,0);

\node[] (a)at(9.5,-6){$\Pi_{[4]}: $};
\node[] (a)at(11.3,-6){$\{\{a_{1},a_{2},a_{3},a_{4}\}\}$};

\node[] (a)at(10.8,-4){$\Pi_{[2,2]}: $};
\node[] (a)at(13,-4){$\{\{a_{1},a_{2}\},\{a_{3},a_{4}\}\}$};
\node[] (a)at(13,-4.5){$\{\{a_{1},a_{3}\},\{a_{2},a_{4}\}\}$};
\node[] (a)at(13,-5){$\{\{a_{1},a_{4}\},\{a_{2},a_{3}\}\}$};

\node[] (a)at(0.8,-4){$\Pi_{[1,3]}: $};
\node[] (a)at(3,-4){$\{\{a_{1}\},\{a_{2},a_{3},a_{4}\}\}$};
\node[] (a)at(3,-4.5){$\{\{a_{2}\},\{a_{1},a_{3},a_{4}\}\}$};
\node[] (a)at(3,-5){$\{\{a_{3}\},\{a_{1},a_{2},a_{4}\}\}$};
\node[] (a)at(3,-5.5){$\{\{a_{4}\},\{a_{1},a_{2},a_{3}\}\}$};

\node[] (a)at(9.0,-2){$\Pi_{[1,1,2]}: $};
\node[] (a)at(11.5,-2){$\{\{a_{1}\},\{a_{2}\},\{a_{3},a_{4}\}\}$};
\node[] (a)at(11.5,-2.5){$\{\{a_{1}\},\{a_{3}\},\{a_{2},a_{4}\}\}$};
\node[] (a)at(11.5,-3){$\{\{a_{1}\},\{a_{4}\},\{a_{2},a_{3}\}\}$};
\node[] (a)at(15.3,-2){$\{\{a_{2}\},\{a_{3}\},\{a_{1},a_{4}\}\}$};
\node[] (a)at(15.3,-2.5){$\{\{a_{2}\},\{a_{4}\},\{a_{1},a_{3}\}\}$};
\node[] (a)at(15.3,-3){$\{\{a_{3}\},\{a_{4}\},\{a_{1},a_{2}\}\}$};

\node[] (a)at(13.4,-2){$,$};
\node[] (a)at(13.4,-2.5){$,$};
\node[] (a)at(13.4,-3){$,$};

\node[] (a)at(9.5,0){$\Pi_{[1,1,1,1]}: $};
\node[] (a)at(12.3,0){{$\{\{a_{1}\},\{a_{2}\},\{a_{3}\},\{a_{4}\}\}$}};

\draw[->,>=latex,vert!75!black!75,very thick] (7.5,-5.7) -- (5.5,-4.3);
\draw[->,>=latex,vert!75!black!75,very thick] (7.5,-5.7) -- (9.5,-4.3);
\draw[->,>=latex,red,dashed,very thick] (5.5,-3.7) -- (7.4,-2.3);
\draw[->,>=latex,vert!75!black!75,very thick] (9.5,-3.7) -- (7.5,-2.3);
\draw[->,>=latex,vert!75!black!75,thick] (7.5,-1.7) -- (7.5,-0.3);

\node[] (a)at(9.3,-5.3){\small$4=2+2$};
\node[] (a)at(9.3,-4.9){\small Split $4$};
\node[] (a)at(5.7,-5.3){\small$4=1+3$};
\node[] (a)at(5.7,-4.9){\small Split $4$};

\node[] (a)at(5.2,-3.3){\small$3=1+2$};
\node[] (a)at(5.2,-2.9){\small Split $3$};
\node[] (a)at(7.95,-3.5){\small$2=1+1$};
\node[] (a)at(7.95,-3.1){\small Split $2$};

\node[] (a)at(6.7,-1.2){\small$2=1+1$};
\node[] (a)at(6.7,-0.8){\small Split $2$};

\end{tikzpicture}%
}
\caption{\small A four-agent integer partition graph.} \label{IPG}
\end{figure}
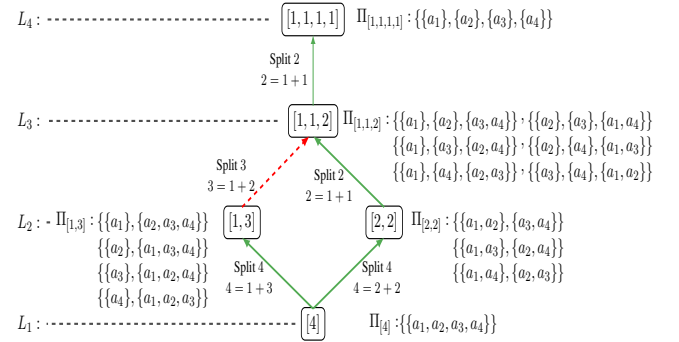

\section{SMART: A Novel CSG Algorithm}

The SMART algorithm is based on three techniques (CDP, GRAD and DIPS) that combine dynamic programming with integer partition graph search. 
SMART introduces new ways of searching the integer partition graph of solutions. 

\subsection{Complementarity-Based Dynamic Programming (CDP)}

CDP is an algorithm that determines the optimal coalition structure. To compute the optimal structure, CDP evaluates different sets of coalitions through two processes (Figure \ref{CDP}), and computes the best partition of each coalition, meaning the best way to split it into potentially multiple subcoalitions. 
The highest valued coalition structure returned by these processes is an optimal solution. 
To determine the coalitions to evaluate (that is, for which to compute the best partitions) 
and 
ensure that the optimal solution is found, the CDP algorithm uses an offline phase of preprocessing. This phase defines the best pair of coalition size sets to evaluate, such that when combined, the entire solution space is searched. This means that by evaluating these specific sets of coalitions, the CDP algorithm can guarantee that it has considered every possible grouping of agents. 

\subsubsection{CDP's Offline Phase}

The offline phase is one of the key components of the CDP algorithm, which is responsible for determining the coalitions to evaluate in the two processes of the algorithm. This is done by considering the coalition sizes, as illustrated on the integer partition graph (Figure \ref{IPG}). 
To understand how this works, let us consider an example of four agents. Dividing a coalition of size 2 into two coalitions of size 1, when searching for the solutions, corresponds to an upward movement in the integer partition graph from the node $[2,2]$ to the node $[1,1,2]$ (2=1+1). By choosing to split in this graph a subset of integers starting from the bottom node, that is, by considering the edges that result from splitting a subset of sizes, a subset of nodes in the integer partition graph becomes reachable from the bottom node, which means that the nodes are connected to the bottom node through a series of edges.
Thus, to search a certain number of subspaces, several sets of coalition sizes could be considered, with a different run time for each set. For example, by splitting only the sizes 2 and 4 starting from the bottom node, that is, by deciding to evaluate all the coalitions of sizes 2 and 4 and not those of size 3, all the nodes in the integer partition graph are reachable from the bottom node through a series of edges that result from splitting the coalition sizes 2 and 4. Hence, the set of sizes $\{2,4\}$ generates $100\%$ of subspaces, as does the set $\{2,3,4\}$, but with a lower run time. 
%
Hence, some sets of sizes may be more beneficial than others. To find the best size set pair, we propose the \textit{Size Sets Definition (SSD)} algorithm. 
The SSD algorithm starts by estimating the time required to evaluate the coalitions of each size from 2 to $n$. This time corresponds to the evaluation time of all the different ways of splitting the coalitions of each size. The estimated time of a splitting is the computational cost associated with this hardware operation, which is a fixed value like any other operation, such as an addition or a subtraction. Then, SSD computes the best pair of coalition size sets that searches all the subspaces with minimum run time. 

Then, SSD goes through each possible pair of coalition size sets, and for each pair, it constructs two integer partition graphs by only dividing the integers that belong to each set of sizes (Figure \ref{CDP}). The edges that result from dividing the sizes of the set connect a number of nodes to the bottom node. 
For a particular pair of sets, in case the generated nodes of the second set of sizes cover all the missed nodes of the first set, meaning that all the solution subspaces are obtained by the first or second set, the SSD algorithm tests whether the pair minimizes the run time. If so, this pair becomes the best pair. The run time of a set of sizes is the sum of the evaluation times of the coalitions of the size set. Given this, the run time of a pair of size sets is the highest run time of the size sets that comprise it. 

Algorithm 1 shows how SSD computes the best pair of coalition size sets. The sets of sizes are represented in binary format, with numbers of $n-2$ bits that represent sizes between 2 and $n-1$. The coalitions of size 1 and $n$ are not evaluated because the size 1 coalitions are never split and the coalition of size $n$ is always split. For example with five agents, the sets of sizes are $\{2,3\}$ and $\{2,4\}$, are represented by the binary numbers $011_{2}$ and $101_{2}$, respectively.

\begin{algorithm}[h!]
\KwIn{A problem size $n$, the required time to evaluate the coalitions of each size.}
\KwOut{The best pair $\mathcal{BS}_{1}, \mathcal{BS}_{2}$ of coalition size sets to search the entire solution subspaces. $\mathcal{BS}_{1}$ and $\mathcal{BS}_{2}$ are in binary format.}
\DontPrintSemicolon
 $\mathcal{BS}_{1} \leftarrow 2^{n-2}-1$ $\triangleright$ Initially, the best sizes include all the coalition sizes, i.e. $\{2,..,n-1\}$\;
 $\mathcal{BS}_{2} \leftarrow 2^{n-2}-1$\;
 $x \leftarrow GeneratedSubspaces(2^{n-2}-1)$ $\triangleright$ $GeneratedSubspaces(2^{n-2}-1)$ returns the set of subspaces for a problem of $n$ agents\;
 $t_{1}^{*} \leftarrow SetTime(2^{n-2}-1)$ $\triangleright$ $SetTime$ returns the run time for a set of sizes. The best time for searching the subspaces is initialized to the run time of considering all the sizes\;
 $t_{2}^{*} \leftarrow SetTime(2^{n-2}-1)$\;
 \For($\triangleright$ $i$ corresponds to a set of coalition sizes represented in binary format){$i= 1$ to $2^{n-2}-1$}{
 $y \leftarrow GeneratedSubspaces(i)$\;
 $v \leftarrow x \setminus y$ $\triangleright$ $v$ contains the missed nodes when considering the set $i$\;
 $t_{1} \leftarrow SetTime(i)$\;
 \If($\triangleright$ the set $i$ has a chance to improve the result){$t_{1}<t_{1}^{*}$ or $t_{1}<t_{2}^{*}$}{
 \For(){$j= i+1$ to $2^{n-2}-1$}{
 $y \leftarrow GeneratedSubspaces(j)$\;
 \If($\triangleright$ the set $j$ generates all the missed nodes of the set $i$){$v \subseteq z$}{
 $t_{2} \leftarrow SetTime(j)$\;
 \If($\triangleright$ the pair of sets $\{i,j\}$ improves the result){$t_{1}<t_{1}^{*}$ and $t_{2} \leq t_{2}^{*}$ or $t_{2}<t_{2}^{*}$ and $t_{1} \leq t_{1}^{*}$}{
 $\mathcal{BS}_{1} \leftarrow y$, $\mathcal{BS}_{2} \leftarrow z$\;
 $t_{1}^{*} \leftarrow t_{1}$, $t_{2}^{*} \leftarrow t_{2}$\;
 }
 }
 }
 }
 }
 Add $n$ to $\mathcal{BS}_{1}$ and $\mathcal{BS}_{2}$ $\triangleright$ $n$ is always considered.\;
 Return $\mathcal{BS}_{1}, \mathcal{BS}_{2}$\;
 \caption{Size Sets Definition (SSD) Algorithm} \label{SSD}
\end{algorithm}

These steps are all executed offline, meaning that we run the SSD algorithm only once  for each problem size $n$ (not once for each problem instance) to set up CDP. For example with ten agents, the best pair of coalition size sets that SSD returns is $\mathcal{BS}_{1} = \{2,4,6,10\}$ and $\mathcal{BS}_{2} = \{2,8,10\}$, which together search all the subspaces. 

\subsubsection{CDP's Online Phase}

The CDP algorithm uses these sets to compute the optimal coalition structure each time a problem instance is to be solved. CDP starts, in a first step, by constructing two tables, the partition table $P_{t}$ that stores the optimal partition of each coalition $\mathcal{C}$ in $P_{t}(\mathcal{C})$ and the value table $V_{t}$ that stores the optimal value of each coalition $\mathcal{C}$ in $V_{t}(\mathcal{C})$. $P_{t}(\mathcal{C})$ and $V_{t}(\mathcal{C})$ are computed for each coalition $\mathcal{C}$ by evaluating all possible ways of splitting $\mathcal{C}$ into two coalitions and checking whether it is beneficial to split it or not. For example, for a coalition of size 4, we evaluate its splitting into a coalition of size 1 and a coalition of size 3 (4=1+3) and into two coalitions of size 2 (4=2+2). This evaluation is done by the two CDP processes, which each consider the coalitions whose sizes belong to the sets returned by SSD. In each process, CDP starts evaluating the smallest coalitions first, as the result of this evaluation is used for evaluating larger coalitions (see Algorithms \ref{CDP1} and \ref{CDP2}). In a second step, each process of CDP computes the best coalition structure, among the searched subspaces, by computing the best partition of the grand coalition $A$. 
Hence, the optimal solution that CDP finds is the highest-valued coalition structure produced by these processes (Figure \ref{CDP}). 

Theorem 1 establishes that when considering any pair of coalition size sets, the presence of a path between each node and the bottom node in one of the two integer partition graphs associated with the respective size sets guarantees finding an optimal coalition structure.

\begin{thm}
When considering any pair of coalition size sets to evaluate, 
if 
there is a path between each node 
and the bottom node of one of the two integer partition graphs generated by the two size sets, CDP will fully search the solution subspaces. Thus it finds an optimal coalition structure.
\end{thm}

\begin{proof}
The splitting operations of CDP are represented with edges in the integer partition graph (Figure \ref{CDP}). An edge that connects two adjacent nodes, and results from splitting an integer $x$ into two, represents the evaluation of all coalitions of size $x$ by CDP. If there is a path between the bottom node of the integer partition graph and a node $\mathcal{N}$, then all the coalitions that need to be split to find the best solution in $\mathcal{N}$ are evaluated. 
As all the nodes are at least connected to one of the bottom nodes of the two integer partition graphs generated by the pair of size sets, CDP fully searches each 
subspace.
\end{proof}

Algorithms \ref{CDP1} and \ref{CDP2} detail the pseudocode of CDP. CDP runs in parallel on the two coalition size sets obtained from the offline phase. In Algorithm \ref{CDP2}, CDP computes the optimal coalition structure that belongs to the subspaces searched considering each set. Then, the optimal solution is the highest valued 
solution of the two (see lines 3-7 of Algorithm \ref{CDP1}).

\begin{algorithm}[t]
\KwIn{Set of all possible coalitions and the value $V_{t}(\mathcal{C})$ of each coalition $\mathcal{C}$. A number of agents $n$. Sets of coalition sizes $\mathcal{BS}_{1}$ and $\mathcal{BS}_{2}$ to consider by CDP.}
\KwOut{An optimal coalition structure $\mathcal{CS}^{*}$ and its value.}
\DontPrintSemicolon
 \nonl $\triangleright$ Begin parallel\;
 \nonl $\triangleright$ CDP runs in parallel given $\mathcal{BS}_{1}$\;
 $\mathcal{CS}^{*}_{1}$, $\mathcal{V}^{*}_{1}$ $\leftarrow$ $Computing\_Optimal\_CS(\mathcal{V}_{t},n,\mathcal{BS}_{1})$\;
 \nonl $\triangleright$ CDP runs in parallel given $\mathcal{BS}_{2}$\;
 $\mathcal{CS}^{*}_{2}$, $\mathcal{V}^{*}_{2}$ $\leftarrow$ $Computing\_Optimal\_CS(\mathcal{V}_{t},n,\mathcal{BS}_{2})$\;
\nonl $\triangleright$ End parallel\;
\uIf{$\mathcal{V}^{*}_{1} > \mathcal{V}^{*}_{2}$}{
$\mathcal{V}^{*} \leftarrow \mathcal{V}^{*}_{1}$, $\mathcal{CS}^{*} \leftarrow \mathcal{CS}^{*}_{1}$\;
}
\Else{
$\mathcal{V}^{*} \leftarrow \mathcal{V}^{*}_{2}$, $\mathcal{CS}^{*} \leftarrow \mathcal{CS}^{*}_{2}$\;
}

 Return $\mathcal{CS}^{*}$, $\mathcal{V}^{*}$\;
 \caption{The CDP algorithm} \label{CDP1}
\end{algorithm}

\begin{algorithm}[b]
\KwIn{The value table $\mathcal{V}_{t}$. A number of agents $n$. Set of coalition sizes $\mathcal{BS}$ to consider by CDP.}
\KwOut{An optimal coalition structure $CS^{*}$ and its value.}
\DontPrintSemicolon
 \For{$s \in \mathcal{BS}$}{
\ForEach{$C \subseteq A$, where $|C| = s$}{
\ForEach{$C_{1}, C_{2} \subseteq C$, where $C_{1} \cup C_{2} = C$ and $C_{1} \cap C_{2} = \emptyset$}{
\If{$V_{t}(C_{1}) + V_{t}(C_{2}) > V_{t}(C)$}{
$V_{t}(C) \leftarrow V_{t}(C_{1}) + V_{t}(C_{2})$\;
$P_{t}(C) \leftarrow \{C_{1},C_{2}\}$\;
}
}
}
}
$\mathcal{CS}^{*} \leftarrow $Partition($A, P_{t}$), $\mathcal{V}^{*} \leftarrow V_{t}(A)$ $\triangleright$ the pseudocode of Partition is in the appendix\;
Return $\mathcal{CS}^{*}$, $\mathcal{V}^{*}$\;
 \caption{Computing Optimal $\mathcal{CS}$} \label{CDP2}
\end{algorithm}

\begin{figure*}[h!]
\begin{center}
\resizebox{0.44\textwidth}{12.3cm}{%
     \begin{subfigure}{0.99\textwidth}
         \centering
         \begin{tikzpicture}


\node[draw,rectangle,rounded corners=4pt] (a)at(7.5,0){$1,1,1,1,1,1,1,1,1,1$};

\node[draw,rectangle,rounded corners=4pt] (a)at(7.5,-1.25){$1,1,1,1,1,1,1,1,2$};

\node[draw,rectangle,rounded corners=4pt] (a)at(5.5,-2.5){$1,1,1,1,1,1,1,3$};

\node[draw,rectangle,rounded corners=4pt] (a)at(9.5,-2.5){$1,1,1,1,1,1,2,2$};

\node[draw,rectangle,rounded corners=4pt] (a)at(4.5,-3.75){$1,1,1,1,1,1,4$};

\node[draw,rectangle,rounded corners=4pt] (a)at(7.5,-3.75){$1,1,1,1,1,2,3$};

\node[draw,rectangle,rounded corners=4pt] (a)at(10.5,-3.75){$1,1,1,1,2,2,2$};

\node[draw,rectangle,rounded corners=4pt] (a)at(2,-5){$1,1,1,1,1,5$};

\node[draw,rectangle,rounded corners=4pt] (a)at(4.75,-5){$1,1,1,1,2,4$};

\node[draw,rectangle,rounded corners=4pt] (a)at(7.5,-5){$1,1,1,1,3,3$};

\node[draw,rectangle,rounded corners=4pt] (a)at(10.25,-5){$1,1,1,2,2,3$};

\node[draw,rectangle,rounded corners=4pt] (a)at(13,-5){$1,1,2,2,2,2$};

\node[draw,rectangle,rounded corners=4pt] (a)at(0.75,-6.25){$1,1,1,1,6$};

\node[draw,rectangle,rounded corners=4pt] (a)at(3,-6.25){$1,1,1,2,5$};

\node[draw,rectangle,rounded corners=4pt] (a)at(5.25,-6.25){$1,1,1,3,4$};

\node[draw,rectangle,rounded corners=4pt] (a)at(7.5,-6.25){$1,1,2,2,4$};

\node[draw,rectangle,rounded corners=4pt] (a)at(9.75,-6.25){$1,1,2,3,3$};

\node[draw,rectangle,rounded corners=4pt] (a)at(12,-6.25){$1,2,2,2,3$};

\node[draw,rectangle,rounded corners=4pt] (a)at(14.25,-6.25){$2,2,2,2,2$};

\node[draw,rectangle,rounded corners=4pt, fill=red!85] (a)at(-0.5,-7.5){$1,1,1,7$};

\node[draw,rectangle,rounded corners=4pt] (a)at(1.5,-7.5){$1,1,2,6$};

\node[draw,rectangle,rounded corners=4pt] (a)at(3.5,-7.5){$1,1,3,5$};

\node[draw,rectangle,rounded corners=4pt] (a)at(5.5,-7.5){$1,2,2,5$};

\node[draw,rectangle,rounded corners=4pt] (a)at(7.5,-7.5){$1,1,4,4$};

\node[draw,rectangle,rounded corners=4pt] (a)at(9.5,-7.5){$1,2,3,4$};

\node[draw,rectangle,rounded corners=4pt] (a)at(11.5,-7.5){$2,2,2,4$};

\node[draw,rectangle,rounded corners=4pt] (a)at(13.5,-7.5){$1,3,3,3$};

\node[draw,rectangle,rounded corners=4pt] (a)at(15.5,-7.5){$2,2,3,3$};

\node[draw,rectangle,rounded corners=4pt] (a)at(0.5,-8.75){$1,1,8$};

\node[draw,rectangle,rounded corners=4pt, fill=red!85] (a)at(2.5,-8.75){$1,2,7$};

\node[draw,rectangle,rounded corners=4pt] (a)at(4.5,-8.75){$1,3,6$};

\node[draw,rectangle,rounded corners=4pt] (a)at(6.5,-8.75){$2,2,6$};

\node[draw,rectangle,rounded corners=4pt] (a)at(8.5,-8.75){$1,4,5$};

\node[draw,rectangle,rounded corners=4pt, fill=red!85] (a)at(10.5,-8.75){$2,3,5$};

\node[draw,rectangle,rounded corners=4pt] (a)at(12.5,-8.75){$2,4,4$};

\node[draw,rectangle,rounded corners=4pt] (a)at(14.5,-8.75){$3,3,4$};

\node[draw,rectangle,rounded corners=4pt] (a)at(2.5,-10){$1,9$};

\node[draw,rectangle,rounded corners=4pt] (a)at(5,-10){$2,8$};

\node[draw,rectangle,rounded corners=4pt] (a)at(7.5,-10){$3,7$};

\node[draw,rectangle,rounded corners=4pt] (a)at(10,-10){$4,6$};

\node[draw,rectangle,rounded corners=4pt] (a)at(12.5,-10){$5,5$};

\node[draw,rectangle,rounded corners=4pt] (a)at(7.5,-11.25){$10$};

\node[fill=black!25] () (a)at(7.5,-12.4){\Large (a)};

\draw[->,>=latex,blue] (7.5,-10.95) -- (7.5,-10.3);

\draw[->,>=latex,blue] (7.5,-10.95) -- (5,-10.3);

\draw[->,>=latex,blue] (7.5,-10.95) -- (2.5,-10.3);

\draw[->,>=latex,blue] (7.5,-10.95) -- (10,-10.3);

\draw[->,>=latex,blue] (7.5,-10.95) -- (12.5,-10.3);


\draw[->,>=latex,blue] (5,-9.7) -- (0.5,-9.05);

\draw[->,>=latex,blue] (10,-9.7) -- (4.5,-9.05);

\draw[->,>=latex,blue] (10,-9.7) -- (8.5,-9.05);
\draw[->,>=latex,blue] (10,-9.7) -- (12.5,-9.05);
\draw[->,>=latex,blue] (10,-9.7) -- (14.5,-9.05);

\draw[->,>=latex,blue] (10,-9.7) -- (6.5,-9.05);


\draw[->,>=latex,blue] (4.5,-8.45) -- (3.5,-7.8);
\draw[->,>=latex,blue] (4.5,-8.45) -- (9.5,-7.8);
\draw[->,>=latex,blue] (4.5,-8.45) -- (13.5,-7.8);

\draw[->,>=latex,blue] (6.5,-8.45) -- (1.5,-7.8);

\draw[->,>=latex,blue] (6.5,-8.45) -- (5.5,-7.8);
\draw[->,>=latex,blue] (6.5,-8.45) -- (11.5,-7.8);
\draw[->,>=latex,blue] (6.5,-8.45) -- (15.5,-7.8);

\draw[->,>=latex,blue] (8.5,-8.45) -- (3.5,-7.8);

\draw[->,>=latex,blue] (8.5,-8.45) -- (5.5,-7.8);

\draw[->,>=latex,blue] (12.5,-8.45) -- (7.5,-7.8);
\draw[->,>=latex,blue] (12.5,-8.45) -- (9.5,-7.8);
\draw[->,>=latex,blue] (12.5,-8.45) -- (11.5,-7.8);

\draw[->,>=latex,blue] (14.5,-8.45) -- (13.5,-7.8);
\draw[->,>=latex,blue] (14.5,-8.45) -- (15.5,-7.8);


\draw[->,>=latex,blue] (1.5,-7.2) -- (0.75,-6.55);

\draw[->,>=latex,blue] (1.5,-7.2) -- (3,-6.55);
\draw[->,>=latex,blue] (1.5,-7.2) -- (7.5,-6.55);
\draw[->,>=latex,blue] (1.5,-7.2) -- (9.75,-6.55);

\draw[->,>=latex,blue] (5.5,-7.2) -- (3,-6.55);

\draw[->,>=latex,blue] (7.5,-7.2) -- (5.25,-6.55);

\draw[->,>=latex,blue] (7.5,-7.2) -- (7.5,-6.55);

\draw[->,>=latex,blue] (9.5,-7.2) -- (5.25,-6.55);

\draw[->,>=latex,blue] (9.5,-7.2) -- (9.75,-6.55);
\draw[->,>=latex,blue] (9.5,-7.2) -- (12,-6.55);

\draw[->,>=latex,blue] (15.5,-7.2) -- (9.75,-6.55);

\draw[->,>=latex,blue] (11.5,-7.2) -- (7.5,-6.55);
\draw[->,>=latex,blue] (11.5,-7.2) -- (12,-6.55);
\draw[->,>=latex,blue] (11.5,-7.2) -- (14.25,-6.55);


\draw[->,>=latex,blue] (0.75,-5.95) -- (2,-5.3);
\draw[->,>=latex,blue] (0.75,-5.95) -- (4.75,-5.3);
\draw[->,>=latex,blue] (0.75,-5.95) -- (7.5,-5.3);

\draw[->,>=latex,blue] (3,-5.95) -- (2,-5.3);

\draw[->,>=latex,blue] (5.25,-5.95) -- (7.5,-5.3);

\draw[->,>=latex,blue] (5.25,-5.95) -- (10.25,-5.3);

\draw[->,>=latex,blue] (7.5,-5.95) -- (4.75,-5.3);

\draw[->,>=latex,blue] (7.5,-5.95) -- (10.25,-5.3);

\draw[->,>=latex,blue] (7.5,-5.95) -- (13,-5.3);

\draw[->,>=latex,blue] (9.75,-5.95) -- (7.5,-5.3);

\draw[->,>=latex,blue] (12,-5.95) -- (10.25,-5.3);

\draw[->,>=latex,blue] (14.25,-5.95) -- (13,-5.3);


\draw[->,>=latex,blue] (4.75,-4.7) -- (4.5,-4.05);

\draw[->,>=latex,blue] (4.75,-4.7) -- (7.5,-4.05);

\draw[->,>=latex,blue] (4.75,-4.7) -- (10.5,-4.05);

\draw[->,>=latex,blue] (10.25,-4.7) -- (7.5,-4.05);

\draw[->,>=latex,blue] (13,-4.7) -- (10.5,-4.05);


\draw[->,>=latex,blue] (4.5,-3.45) -- (5.5,-2.8);

\draw[->,>=latex,blue] (4.5,-3.45) -- (9.5,-2.8);

\draw[->,>=latex,blue] (7.5,-3.45) -- (5.5,-2.8);

\draw[->,>=latex,blue] (10.5,-3.45) -- (9.5,-2.8);


\draw[->,>=latex,blue] (9.5,-2.2) -- (7.5,-1.55);

\draw[->,>=latex,blue] (7.5,-0.95) -- (7.5,-0.3);

\end{tikzpicture}\\
         \label{fig:three sin x}
     \end{subfigure}
}
\:\:\:\:\:\:\:\:\:\:\:\:
\resizebox{0.44\textwidth}{12.3cm}{%
     \begin{subfigure}{0.99\textwidth}
         \centering
         \begin{tikzpicture}


\node[draw,rectangle,rounded corners=4pt, fill=red!85] (a)at(7.5,0){$1,1,1,1,1,1,1,1,1,1$};

\node[draw,rectangle,rounded corners=4pt, fill=red!85] (a)at(7.5,-1.25){$1,1,1,1,1,1,1,1,2$};

\node[draw,rectangle,rounded corners=4pt, fill=red!85] (a)at(5.5,-2.5){$1,1,1,1,1,1,1,3$};

\node[draw,rectangle,rounded corners=4pt, fill=red!85] (a)at(9.5,-2.5){$1,1,1,1,1,1,2,2$};

\node[draw,rectangle,rounded corners=4pt, fill=red!85] (a)at(4.5,-3.75){$1,1,1,1,1,1,4$};

\node[draw,rectangle,rounded corners=4pt, fill=red!85] (a)at(7.5,-3.75){$1,1,1,1,1,2,3$};

\node[draw,rectangle,rounded corners=4pt, fill=red!85] (a)at(10.5,-3.75){$1,1,1,1,2,2,2$};

\node[draw,rectangle,rounded corners=4pt, fill=red!85] (a)at(2,-5){$1,1,1,1,1,5$};

\node[draw,rectangle,rounded corners=4pt, fill=red!85] (a)at(4.75,-5){$1,1,1,1,2,4$};

\node[draw,rectangle,rounded corners=4pt, fill=red!85] (a)at(7.5,-5){$1,1,1,1,3,3$};

\node[draw,rectangle,rounded corners=4pt, fill=red!85] (a)at(10.25,-5){$1,1,1,2,2,3$};

\node[draw,rectangle,rounded corners=4pt, fill=red!85] (a)at(13,-5){$1,1,2,2,2,2$};

\node[draw,rectangle,rounded corners=4pt] (a)at(0.75,-6.25){$1,1,1,1,6$};

\node[draw,rectangle,rounded corners=4pt, fill=red!85] (a)at(3,-6.25){$1,1,1,2,5$};

\node[draw,rectangle,rounded corners=4pt, fill=red!85] (a)at(5.25,-6.25){$1,1,1,3,4$};

\node[draw,rectangle,rounded corners=4pt, fill=red!85] (a)at(7.5,-6.25){$1,1,2,2,4$};

\node[draw,rectangle,rounded corners=4pt, fill=red!85] (a)at(9.75,-6.25){$1,1,2,3,3$};

\node[draw,rectangle,rounded corners=4pt, fill=red!85] (a)at(12,-6.25){$1,2,2,2,3$};

\node[draw,rectangle,rounded corners=4pt, fill=red!85] (a)at(14.25,-6.25){$2,2,2,2,2$};

\node[draw,rectangle,rounded corners=4pt] (a)at(-0.5,-7.5){$1,1,1,7$};

\node[draw,rectangle,rounded corners=4pt] (a)at(1.5,-7.5){$1,1,2,6$};

\node[draw,rectangle,rounded corners=4pt] (a)at(3.5,-7.5){$1,1,3,5$};

\node[draw,rectangle,rounded corners=4pt, fill=red!85] (a)at(5.5,-7.5){$1,2,2,5$};

\node[draw,rectangle,rounded corners=4pt] (a)at(7.5,-7.5){$1,1,4,4$};

\node[draw,rectangle,rounded corners=4pt, fill=red!85] (a)at(9.5,-7.5){$1,2,3,4$};

\node[draw,rectangle,rounded corners=4pt, fill=red!85] (a)at(11.5,-7.5){$2,2,2,4$};

\node[draw,rectangle,rounded corners=4pt, fill=red!85] (a)at(13.5,-7.5){$1,3,3,3$};

\node[draw,rectangle,rounded corners=4pt, fill=red!85] (a)at(15.5,-7.5){$2,2,3,3$};

\node[draw,rectangle,rounded corners=4pt] (a)at(0.5,-8.75){$1,1,8$};

\node[draw,rectangle,rounded corners=4pt] (a)at(2.5,-8.75){$1,2,7$};

\node[draw,rectangle,rounded corners=4pt, fill=red!85] (a)at(4.5,-8.75){$1,3,6$};

\node[draw,rectangle,rounded corners=4pt] (a)at(6.5,-8.75){$2,2,6$};

\node[draw,rectangle,rounded corners=4pt, fill=red!85] (a)at(8.5,-8.75){$1,4,5$};

\node[draw,rectangle,rounded corners=4pt] (a)at(10.5,-8.75){$2,3,5$};

\node[draw,rectangle,rounded corners=4pt] (a)at(12.5,-8.75){$2,4,4$};

\node[draw,rectangle,rounded corners=4pt, fill=red!85] (a)at(14.5,-8.75){$3,3,4$};

\node[draw,rectangle,rounded corners=4pt] (a)at(2.5,-10){$1,9$};

\node[draw,rectangle,rounded corners=4pt] (a)at(5,-10){$2,8$};

\node[draw,rectangle,rounded corners=4pt] (a)at(7.5,-10){$3,7$};

\node[draw,rectangle,rounded corners=4pt] (a)at(10,-10){$4,6$};

\node[draw,rectangle,rounded corners=4pt] (a)at(12.5,-10){$5,5$};

\node[draw,rectangle,rounded corners=4pt] (a)at(7.5,-11.25){$10$};

\node[fill=black!25] () (a)at(7.5,-12.4){\Large (b)};

\draw[->,>=latex,blue] (7.5,-10.95) -- (7.5,-10.3);

\draw[->,>=latex,blue] (7.5,-10.95) -- (5,-10.3);

\draw[->,>=latex,blue] (7.5,-10.95) -- (2.5,-10.3);

\draw[->,>=latex,blue] (7.5,-10.95) -- (10,-10.3);

\draw[->,>=latex,blue] (7.5,-10.95) -- (12.5,-10.3);


\draw[->,>=latex,blue] (5,-9.7) -- (0.5,-9.05);

\draw[->,>=latex,vert] (5,-9.7) -- (2.5,-9.05);
\draw[->,>=latex,vert] (5,-9.7) -- (6.5,-9.05);
\draw[->,>=latex,vert] (5,-9.7) -- (10.5,-9.05);
\draw[->,>=latex,vert] (5,-9.7) -- (12.5,-9.05);


\draw[->,>=latex,vert] (0.5,-8.45) -- (-0.5,-7.8);
\draw[->,>=latex,vert] (0.5,-8.45) -- (1.5,-7.8);
\draw[->,>=latex,vert] (0.5,-8.45) -- (3.5,-7.8);
\draw[->,>=latex,vert] (0.5,-8.45) -- (7.5,-7.8);

\draw[->,>=latex,blue] (2.5,-8.45) -- (-0.5,-7.8);

\draw[->,>=latex,blue] (6.5,-8.45) -- (1.5,-7.8);

\draw[->,>=latex,blue] (10.5,-8.45) -- (3.5,-7.8);

\draw[->,>=latex,blue] (12.5,-8.45) -- (7.5,-7.8);


\draw[->,>=latex,blue] (1.5,-7.2) -- (0.75,-6.55);


\end{tikzpicture}\\
         \label{fig:three sin x}
     \end{subfigure}
     }

             \caption{Illustration of CDP on a 10-agent integer partition graph. CDP evaluates the coalitions of size {\small$s$ $\in \mathcal{BS}_{1}$=$\{2,4,6,10\}$} (Figure \ref{CDP}.a) and those of size {\small$s$ $\in \mathcal{BS}_{2}$=$\{2,8,10\}$} (Figure \ref{CDP}.b) in parallel. With the set {\small$\mathcal{BS}_{1}$} (resp. {\small$\mathcal{BS}_{2}$}), CDP explores all the subspaces in Figure \ref{CDP}.a (resp. Figure \ref{CDP}.b), except the red ones. Nevertheless, CDP covers all the subspaces  using {\small$\mathcal{BS}_{1}$} and {\small$\mathcal{BS}_{2}$}. Hence, no node is missed by both sets, that is, no node is red in both figures.
             }
        \label{CDP}
\normalsize
\end{center}
\end{figure*}

\subsection{Gradual Search with Dynamic Programming (GRAD)}

The GRAD algorithm uses multiple parallel processes to search for the optimal solution, each with a set of coalition sizes as input with which it explores a certain percentage of the search space. These percentages that we detail in Section~6 are hyperparameters that can be adjusted to fine-tune the algorithm. 

\subsubsection{GRAD's Offline Phase}

GRAD also uses an offline phase to compute, for each considered percentage $\omega$, the best coalition size set that allows one to search this percentage of subspaces with the shortest run time. 
To find these sets, we introduce 
the \textit{Size Optimization for diFferent percenTages (SOFT)} algorithm. For each size set $\mathcal{S}$, SOFT constructs the corresponding integer partition graph $\mathcal{G}_{\mathcal{S}}$ by only dividing the integers that belong to the set $\mathcal{S}$. $\mathcal{G}_{\mathcal{S}}$ is thus partial, as shown in the example of Figure~\ref{CDP}. 
If this number of generated nodes in $\mathcal{G}_{\mathcal{S}}$ is at least an $\omega$ fraction of the total number of subspaces and $\mathcal{S}$ minimizes the run time, then $\mathcal{S}$ becomes the best set. 
Algorithm \ref{soft} shows how SOFT computes the best coalition size sets. 
For example with $n=10$ agents and $\omega = 90\%$, the best coalition size set that SOFT returns is $\{2,4,6,10\}$; it searches $92.86\%$ of subspaces (Figure \ref{CDP}.a). 

\begin{algorithm}[h!]
\KwIn{A CSG problem size $n$, the required times to evaluate the coalitions of sizes 2 to $n$-1 for splitting, and the percentage $\omega$ of solution subspaces.}
\KwOut{The best set of coalition sizes $\mathcal{BS}$ to search at least the percentage $\omega$ of the solution subspaces. $\mathcal{BS}$ is in binary format.}
\DontPrintSemicolon
 $\mathcal{BS} \leftarrow 2^{n-2}-1$ $\triangleright$ Initially, the best set includes all the coalition sizes, i.e. $\{2,..,n-1\}$. $2^{n-2}-1$ is the binary representation of this set.\;
 $x \leftarrow NumberOfSubspaces(2^{n-2}-1)$ $\triangleright$ The function$NumberOfSubspaces$ returns the number of subspaces for a CSG problem with $n$ agents\;
 $t^{*} \leftarrow SetTime(2^{n-2}-1)$ $\triangleright$ $SetTime$ returns the run time for a set of sizes. $t^{*}$ is initialized to the run time when considering all the coalition sizes, which is the worst time\;
 \For($\triangleright$ $i$ corresponds to a set of coalition sizes represented in binary format){$i= 1$ to $2^{n-2}-1$}{
 $y \leftarrow NumberOfSubspaces(i)$ $\triangleright$ $y$ represents the number of generated subspaces in $i$\;
 \If($\triangleright$ the set $i$ generates at least the needed percentage of subspaces){$y\geq x \times \omega$ \textbf{and} $SetTime(i) < t^{*}$}{
 $\mathcal{BS} \leftarrow i$\;
 $t^{*} \leftarrow SetTime(i)$\;
 }
 }
 Add $n$ to $\mathcal{BS}$ $\triangleright$ $n$ is always considered.\;
 Return $\mathcal{BS}$\;
 \caption{The SOFT Algorithm} \label{soft}
\end{algorithm}

\subsubsection{GRAD's Online Phase}

Once these size sets are computed by the SOFT algorithm, each process of GRAD is tuned with the corresponding size set. To solve the problems, GRAD builds a partial integer partition graph with all subspaces and no edges and launches each process with its size set and this partial graph as input (Algorithm \ref{grad}). Each process of GRAD evaluates all the coalitions whose sizes belong to the best size set obtained from the offline phase and computes their best partitions. The GRAD process evaluates all possible ways of splitting each coalition of the selected sizes into two coalitions and tests whether it is beneficial to split or not. The coalitions are evaluated starting with the smallest ones (Figure \ref{CDP}.a). The result of this evaluation is stored in the partition table $P_{t}$ and the value table $V_{t}$. Once all the coalitions have been evaluated, the GRAD process returns the best coalition structure among the searched subspaces. 
This is determined by computing the best partition of the grand coalition $A$ using the partition and value tables generated during the evaluation process.

When the optimal solution is in the subspaces explored by a process of GRAD that searches a specific percentage of subspaces $\omega$ $<$ $100\%$, the process 
finds it with the shortest run time and enables the other GRAD processes to instantly prune certain subspaces without exploring them. To prune the subspaces, we introduce the upper bound $\mathcal{UB}(\mathcal{N})$ of a subspace $\mathcal{N}$, which is the highest value a coalition structure of this subspace can possibly reach. $\mathcal{UB} (\mathcal{N})$ = $\sum_{i \in Integers(\mathcal{N})}$ $Max_{i}$, where $Max_{i}$ is the maximum value a coalition of size $i$ can take and $Integers(\mathcal{N})$ is the set of integers that form the corresponding integer partition of the subspace $\mathcal{N}$. For instance, for $\mathcal{N}$ = $[1,4,5]$, $Integers(\mathcal{N})$ = $\{1,4,5\}$ and $\mathcal{UB} (\mathcal{N})$ = $\sum_{i \in \{1,4,5\}} Max_{i}$ $=$ $Max_{1}$ $+$ $Max_{4}$ $+$ $Max_{5}$. 
By comparing the upper bounds of the subspaces, the GRAD processes identify those that have no chance of improving the solution quality and prune the corresponding nodes that do not have a better upper bound than the last best solution found (Line 10 in Algorithm \ref{gradp}). Moreover, after evaluating all the coalitions of size $x$, a GRAD process updates the integer partition graph by adding the edges that result from splitting $x$ into two integers (Line 8 in Algorithm \ref{gradp}). For example in Figure \ref{CDP}.b, after evaluating all the coalitions of size 8, all the green edges are added to the graph. Hence, a number of nodes become reachable from the bottom node through a series of edges, and the corresponding subspaces are fully searched. Thus, the GRAD process prunes them (Line 9 in Algorithm \ref{gradp}). 
The subspace pruning, using the upper bounds and the connection to the bottom node, are repeated each time a process of GRAD finishes evaluating the coalitions of each size. Hence, as the size $x$ increases, more subspaces are pruned from the graph. Thus, when all subspaces are pruned, GRAD finishes and returns the optimal solution.

\begin{algorithm}[h!]
\KwIn{Set of all possible coalitions and the value $V_{t}(\mathcal{C})$ and $P_{t}(\mathcal{C})$ of each coalition $\mathcal{C}$. A number of agents $n$. Sets of coalition sizes $\mathcal{BS}_{i}$ to consider by GRAD.}
\KwOut{An optimal coalition structure $\mathcal{CS}^{*}$ and its value.}
\DontPrintSemicolon
\For($\triangleright$ $|\mathcal{BS}|$ is the number of considered percentages.){$i = 1$ \textbf{to} $|\mathcal{BS}|$}{
 Generate a partial integer partition graph $\mathcal{I}_{\mathcal{G}}$ with all subspaces and with no edges\;
 \nonl $\triangleright$ Begin parallel\;
 \nonl $\triangleright$ GRAD process runs in parallel given $\mathcal{BS}_{i}$\;
 $\mathcal{CS}^{*}$, $\mathcal{V}^{*}$ $\leftarrow$ $Search\_Process(\mathcal{P}_{t},\mathcal{V}_{t},n,\mathcal{BS}_{i},\mathcal{I}_{\mathcal{G}})$\;
 \nonl $\triangleright$ End parallel\;
}
 Return $\mathcal{CS}^{*}$, $\mathcal{V}^{*}$\;
 \caption{The GRAD algorithm} \label{grad}
\end{algorithm}

\begin{algorithm}[h!]
\KwIn{The partition and value tables $\mathcal{P}_{t}$ and $\mathcal{V}_{t}$. A number of agents $n$. Set of coalition sizes $\mathcal{BS}$ to consider by the process. A partial integer partition graph.}
\KwOut{An optimal coalition structure $CS^{*}$ and its value.}
\DontPrintSemicolon
 \For{$s \in \mathcal{BS}$}{
\ForEach{$C \subseteq A$, where $|C| = s$}{
\ForEach{$C_{1}, C_{2} \subseteq C$, where $C_{1} \cup C_{2} = C$ and $C_{1} \cap C_{2} = \emptyset$}{
\If{$V_{t}(C_{1}) + V_{t}(C_{2}) > V_{t}(C)$}{
$V_{t}(C) \leftarrow V_{t}(C_{1}) + V_{t}(C_{2})$\;
$P_{t}(C) \leftarrow \{C_{1},C_{2}\}$\;
}
}
}
Compute the best solution, {\scriptsize$\mathcal{CS}^{*}$, $\mathcal{V}^{*}$}, from $P_{t}$ and $V_{t}$\;
Add to the integer partition graph the edges that result from splitting the size $s$\;
Prune the subspaces connected to the bottom node $\triangleright$ these subspaces are already explored\;
Prune the subspaces that do not have a better upper bound than the last best solution found\;
\If{all the subspaces are pruned}{
Return $\mathcal{CS}^{*}$, $\mathcal{V}^{*}$\;
}
}
 \caption{Search Process} \label{gradp}
\end{algorithm}

\subsection{Distributed Integer Partition Graph Search (DIPS)}

DIPS searches the solution subspaces using the integer partition graph. First, DIPS computes the upper bounds of the subspaces and searches them based on their upper bounds. Then, whenever a CDP or GRAD process finishes, while there are still unexplored nodes, DIPS 
uses that process for a different problem space to parallelize its search. The new process uses the same search technique in DIPS. 
Thus, the subspaces of solutions will gradually be distributed between several processes as they are released by CDP or GRAD, which share the subspaces, their upper bounds, and their sorting. This way, each subspace is searched by only one process. 

DIPS progressively prunes the subspaces that do not have a better upper bound than the last best solution found. To search a subspace of solutions, a DIPS process constructs several search trees to explore the coalition structures. The nodes of these trees represent coalitions and each path from the root to a leaf represents a coalition structure. Moreover, DIPS applies a branch-and-bound technique to identify and avoid branches that have no chance of containing an optimal solution. An example of this step is given in Figure 3 in the appendix.

\section{Hybridization: The SMART Algorithm}

We combine CDP, GRAD, and DIPS to make the \textit{coalition-\textbf{S}ize opti\textbf{M}ization and subsp\textbf{A}ce \textbf{R}econfigura\textbf{T}ion (SMART)} algorithm. Initially, SMART sorts the subspaces by their upper bounds. The DIPS algorithm starts searching with the subspace that has the highest upper bound. Then,  DIPS prunes out the subspaces that are either already searched by CDP or GRAD, or that do not have a better upper bound than the last best solution found. CDP and GRAD evaluate the coalitions of the computed sizes obtained from their respective offline phases and allow subspace pruning through intermediate solutions. Whenever a process in CDP or GRAD finishes evaluating the coalitions of any size, they prune out the subspaces that are connected to the bottom node of the integer partition graph through a series of edges because the optimal coalition structure among these subspaces is found by CDP or GRAD. Hence, DIPS does not need to search them. Figure \ref{dips2} shows how DIPS distributes the search.

\begin{figure}[h!]
\centering
\resizebox{1.0\columnwidth}{2.6cm}{%
\begin{tikzpicture}
\Large

\node[draw,rectangle,rounded corners=5pt, minimum height=1cm, minimum width=1.7cm, text width=1.6cm, text centered] (a)at(5.0,2.5){DIPS process};

\draw[->,>=latex,red,dashed, thick, line width=1pt] (5.0,1.975) -- (1.5,0.025);

\draw[->,>=latex,red,dashed, thick, line width=1pt] (5.0,1.975) -- (4.0,0.025);

\draw[->,>=latex,red,dashed, thick, line width=1pt] (5.0,1.975) -- (6.5,0.025);

\draw[->,>=latex,red,dashed, thick, line width=1pt] (5.0,1.975) -- (11.5,0.025);

\draw[->,>=latex,red,dashed, thick, line width=1pt] (5.0,1.975) -- (19,0.025);

\node[draw,rectangle,rounded corners=5pt, minimum height=1cm, minimum width=1.9cm, text width=1.9cm, text centered] (a)at(10.0,2.5){GRAD process 1};

\draw[->,>=latex,linecolor,dashed, thick, line width=1pt] (10.0,1.975) -- (9.0,0.025);

\draw[->,>=latex,linecolor,dashed, thick, line width=1pt] (10.0,1.975) -- (16.5,0.025);

\node[draw,rectangle,rounded corners=5pt, minimum height=1cm, minimum width=1.9cm, text width=1.9cm, text centered] (a)at(15.0,2.5){GRAD process 2};

\draw[->,>=latex,vert,dashed, thick, line width=1pt] (15.0,1.975) -- (14,0.025);

\node[draw,rectangle,rounded corners=5pt, minimum height=1cm, minimum width=1.7cm, text width=1.6cm, text centered] (a)at(1.5,-0.5){2,4,4};

\node[draw,rectangle,rounded corners=5pt, minimum height=1cm, minimum width=1.7cm, text width=1.6cm, text centered] (a)at(4.0,-0.5){1,2,2,6};

\node[draw,rectangle,rounded corners=5pt, minimum height=1cm, minimum width=1.7cm, text width=1.6cm, text centered] (a)at(6.5,-0.5){3,3,4};

\node[draw,rectangle,rounded corners=5pt, minimum height=1cm, minimum width=1.7cm, text width=1.6cm, text centered] (a)at(9.0,-0.5){1,1,1,3,4};

\node[draw,rectangle,rounded corners=5pt, minimum height=1cm, minimum width=1.7cm, text width=1.6cm, text centered] (a)at(11.5,-0.5){1,3,6};

\node[draw,rectangle,rounded corners=5pt, minimum height=1cm, minimum width=1.7cm, text width=1.6cm, text centered] (a)at(14.0,-0.5){2,2,3,3};

\node[draw,rectangle,rounded corners=5pt, minimum height=1cm, minimum width=1.7cm, text width=1.6cm, text centered] (a)at(16.5,-0.5){1,2,2,5};
 
\node[draw,rectangle,rounded corners=5pt, minimum height=1cm, minimum width=1.7cm, text width=1.6cm, text centered] (a)at(19.0,-0.5){1,4,5};

\end{tikzpicture}%
}
\caption{Illustration of the subspace distribution technique. The subspaces are sorted according to their upper bounds. The subspace [2,4,4] is the highest upper bound node and [1,4,5] is the lowest upper bound node. First, DIPS starts by searching the highest upper bound subspaces. Then, each time a GRAD or  CDP process is released, DIPS uses the process to expand the parallelism of its search. For example in this Figure, when the released ``GRAD process 1'' is used by DIPS to search the node [1,1,1,3,4], the ``DIPS process'' searches another node. 
}\label{dips2}
\end{figure}
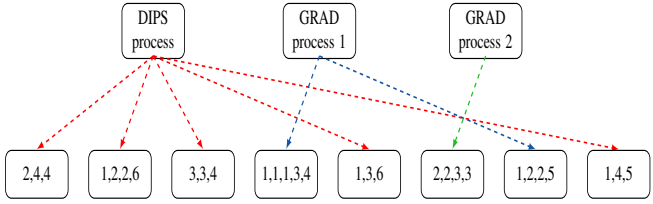

Algorithm~\ref{smartcode} shows the pseudocode of SMART. We now introduce the following results.

\begin{lm}
CDP is faster than or at least as fast as IDP.
\end{lm}
\begin{proof}
Let $\mathcal{S}_{IDP}$ be the set of sizes used by IDP~\cite{rahwan2008improved}. $\mathcal{S}_{IDP}$ is hand tuned and is always equal to $\{2,3,..,\frac{2n}{3},n\}$ for $n$ agents. Let $\mathcal{S}_{1}$ and $\mathcal{S}_{2}$ be the sets of sizes used by CDP, which are tuned automatically by the SSD algorithm. $\mathcal{S}_{1}$ and $\mathcal{S}_{2}$ are configured to be the pair of sets with the shortest resulting run time. We now show that IDP can not have a better run time than CDP. Let $time(\mathcal{S}_{i})$ be the run time produced by the set $i$. Suppose now that there exists a number of agents for which $\mathcal{S}_{IDP}$ is the best set of sizes to consider and that no pair of sets can produce a better run time. In this case, as the sets $\mathcal{S}_{1}$ and $\mathcal{S}_{2}$ find the shortest run time, they must coincide with $\mathcal{S}_{IDP}$, meaning that $max(time(\mathcal{S}_{1}),time(\mathcal{S}_{2}))=time(\mathcal{S}_{IDP})$. Hence, the SSD algorithm would find the sets $\mathcal{S}_{1} =\mathcal{S}_{IDP}$ and $\mathcal{S}_{2}$ is not necessary because $\mathcal{S}_{IDP}$ is sufficient to search the entire solution space and the statement follows.
\end{proof}

Lemma~1 shows a notable property of CDP, and hence of SMART. In particular, it enables us to prove the following theorem. To the best of our knowledge, ODP-IP~\cite{michalak2016hybrid}, ODSS~\cite{changder2020odss}, and BOSS~\cite{Changder_Aknine_Ramchurn_Dutta_2021} are the fastest prior optimal algorithms for the CSG problem. 
\begin{thm}
In the worst case, SMART is faster than or at least as fast as ODP-IP, ODSS, and BOSS.
\end{thm}
\begin{proof}
SMART uses the CDP algorithm, which relies on having the best pair of coalition size sets that enables fast search of optimal results. By Lemma~1, CDP is faster than IDP. 
Recall that the fastest exact algorithms are the hybrid solutions, ODP-IP, ODSS, and BOSS, that combine IDP and an integer partition-based algorithm. However, this combination is highly dependent on the efficiency of the integer partition-based algorithm, which in the worst case requires searching all the coalition structures in $O(n^{n})$ time~\cite{rahwan2009anytime}, which is infeasible in a reasonable time. Thus, in the worst case, the run time of such hybrid algorithms is determined by the dynamic programming approach. Hence, for hard problems, SMART represents the fastest solution as the dynamic programming algorithm used (CDP) is faster than IDP used by the other algorithms. 
Formally, let $T_{SMART}$ be the time complexity of SMART, where $T_{SMART} = min(O(n^{n}), time(GRAD), time(CDP)) \leq time(CDP)$ and let $T_{Other}$ be the time complexity of the other algorithms, where $T_{Other} = min(O(n^{n}), time(IDP)) = time(IDP)$. Given that $time(CDP) \leq time(IDP)$ by Lemma~1, $T_{SMART} \leq T_{Other}$. 
\end{proof}

The result of this hybridization is threefold: (1) CDP is faster than the 
dynamic programming algorithm IDP (results reported in Section~6). 
This allows SMART to be faster in the worst case than the state-of-the-art algorithms ODP-IP~\cite{michalak2016hybrid} and BOSS~\cite{Changder_Aknine_Ramchurn_Dutta_2021} because the CDP part of SMART is faster than the IDP algorithm used by ODP-IP and BOSS, and the worst case time performance of these algorithms is determined by their dynamic programming parts; (2) GRAD gradually searches the best size sets for each percentage of solution subspaces. This allows SMART to reach the number of subspaces needed to guarantee finding an optimal solution with the best run time; (3) The integer partitions are distributed among several processes, enabling  an efficient and faster search in the integer partition graph. 

\begin{algorithm}[h!]
\KwIn{Set of all possible coalitions and the value $v(C)$ of each coalition $\mathcal{C}$ for $n$ agents. 
Sets of coalition sizes $\mathcal{BS}_{1}$, $\mathcal{BS}_{2}$ to consider by CDP and $\mathcal{BS}_{i}$ to consider by GRAD.}
\KwOut{An optimal coalition structure $CS^{*}$ and its value.}
\DontPrintSemicolon
 Generate a partial integer partition graph with all subspaces and with no edges\;
 Sort the subspaces by their upper bounds\;
 \nonl $\triangleright$ Begin parallel\;
 \nonl $\triangleright$ CDP runs in parallel with DIPS and GRAD\;
 $\mathcal{CS}^{*}$, $\mathcal{V}^{*}$ $\leftarrow$ CDP$(v,n,\mathcal{BS}_{1},\mathcal{BS}_{2})$\;
 Return $\mathcal{CS}^{*}$, $\mathcal{V}^{*}$\;
 \nonl $\triangleright$ GRAD runs in parallel with DIPS and CDP\;
 $\mathcal{CS}^{*}$, $\mathcal{V}^{*}$ $\leftarrow$ GRAD$(v,n,\mathcal{BS}_{i})$\;
 Return $\mathcal{CS}^{*}$, $\mathcal{V}^{*}$\;
 \nonl $\triangleright$ DIPS runs in parallel with CDP and GRAD\;
 \ForEach{promising subspace $\mathcal{PS}$}{
 DIPS searches the subspace $\mathcal{PS}$\;
 }
 Return $\mathcal{CS}^{*}$, $\mathcal{V}^{*}$\;
 \nonl $\triangleright$ End parallel\;
 \caption{The SMART algorithm} \label{smartcode}
\end{algorithm}

\section{Analysis of SMART}

In this section, we prove that the SMART algorithm is complete in Theorem 3. Then, we analyze the computational complexity of the algorithms in detail.

\begin{thm}
The SMART algorithm always finds the optimal solution.
\end{thm}
\begin{proof}
Each node in the integer partition graph is searched by SMART using the CDP, GRAD, and DIPS algorithms. A node represents a subspace, which contains a number of coalition structures that match the parts of the subspace. The SMART algorithm returns the final solution when all nodes have been searched or pruned. For a particular node that contains the optimal coalition structure, the only way for SMART not to search it is for DIPS or GRAD to prune it without any of the three algorithms searching it. However, DIPS or GRAD will only prune nodes that have no chance of containing the optimal solution. Thus, such a node would never be pruned, and one of the three algorithms would always completely search the node that contains an optimal solution.
\end{proof}

\noindent \textbf{Time complexity of the SSD algorithm.}\: 
For each problem size $n$, the SSD algorithm tests all possible pairs of coalition size sets. For each set of coalition sizes, SSD reconstructs the integer partition graph by dividing only the integers that belong to that set and tests whether it is beneficial to pair that set with another set or not. 
To test this for one set, CDP pairs it with all the other sets. The total number of sets that SSD evaluates is $2^{n-2}-1$. We denote by $\mathcal{I}(s)$ the number of integer splits performed for a set $s$ and by $p(n,s)$, the number of subspaces generated by the set $s$. 
The total number of operations performed by SSD for one set is 
$\mathcal{T}(n) = \sum_{s=1}^{2^{n-2}-1} \mathcal{I}(s) = \sum_{s=1}^{2^{n-2}-1} \sum_{\mathcal{N}} \mathcal{I}(\mathcal{N},s)$, where $\mathcal{I}(\mathcal{N},s)$ is the number of integer splits performed on node $\mathcal{N}$. 
The number of splits into two for a certain integer $i$ is $\frac{i}{2}$, the highest integer to split is $n$, and the highest possible number of integers in a single node is $n$, which is the number of integers of the node that represents the singleton coalition structure. Thus, $\mathcal{I}(\mathcal{N},s) \leq n \times \frac{n}{2}$, and 
$\mathcal{T}(n) \leq \sum_{s=1}^{2^{n-2}-1} p(n,s) \times n \times \frac{n}{2}$. However, the growth rate of the number of nodes in the integer partition graph, which is the same as the growth rate of integer partitions of $n$, is  $\mathcal{O}(\frac{e^{\pi\sqrt{\frac{2n}{3}}}}{n})$~\cite{herbert2000lectures}. Hence,\\
$\mathcal{T}(n) \leq \sum_{s=1}^{2^{n-2}-1} \mathcal{O}(\frac{e^{\pi\sqrt{\frac{2n}{3}}}}{n}) \times n \times \frac{n}{2}$ $ \leq (2^{n-2}-1) \times \mathcal{O}(\frac{e^{\pi\sqrt{\frac{2n}{3}}}}{n}) \times n \times \frac{n}{2}$.\\ 
As a result, the total number of operations of SSD when testing one set is 
$\mathcal{O}(n^{2} \times 2^{n} \times \frac{e^{\pi\sqrt{\frac{2n}{3}}}}{n})$. SSD tests $2^{n-2}-1$ different sets. Thus, the time complexity of SSD is
$\mathcal{O}(2^{2n} \times n^{2} \times \frac{e^{\pi\sqrt{\frac{2n}{3}}}}{n})$.


\noindent \textbf{Time complexity of the SMART algorithm.}\: 
SMART combines three algorithms--CDP, GRAD, and DIPS, and runs them in parallel. The worst-case run time of dynamic programming on this problem is $\mathcal{O}(3^{n})$~\cite{yeh1986dynamic}. CDP and GRAD are based on dynamic programming. They run in parallel on several sets of sizes and terminate when all sets are fully evaluated. Thus, the time complexity of both CDP and GRAD is $\mathcal{O}(3^{n})$. 
The worst-case run time of DIPS, which, in the worst-case, requires us to search all coalition structures, is $\mathcal{O}(n^{n})$. As a result, the time complexity of SMART is $\min(\mathcal{O}(3^{n}),\mathcal{O}(3^{n}),\mathcal{O}(n^{n})) = \mathcal{O}(3^{n})$.


\section{Empirical Evaluation}

We now evaluate the effectiveness of SMART 
by comparing it to the prior state-of-the-art algorithms ODP-IP and BOSS. We implemented SMART in Java and for ODP-IP and BOSS, we used the codes provided by their authors for the comparisons. They are also written in Java. The algorithms were run on an Intel Xeon 2.30GHz E5-2650 CPU with 256GB of RAM. For GRAD, we considered values of $\omega$ $\in \{10\%,20\%,30\%,40\%,50\%,60\%,70\%,80\%,90\%,100\%\}$. 
We also designed and tested a different version of ODP-IP, namely POI (Parallel ODP-IP), that we developed to integrate parallelism in the baseline version of ODP-IP in order to improve its performance. It uses the same number of processes as SMART (see the appendix for more details). This does not affect the theoretical guarantees but improves the practical performances of the algorithm.

We conducted the experiments on common benchmark problems. We show results on nine value distributions. Results on other distributions are in the appendix. We compared the algorithms using the following value distributions: Modified Normal~\cite{rahwan2012hybrid}, Beta, Exponential, Gamma~\cite{michalak2016hybrid}, Normal~\cite{rahwan2007near}, Uniform~\cite{larson2000anytime}, Modified Uniform~\cite{service2010approximate}, Zipf, SVA Beta and Weibull~\cite{changder2020odss}. The experiments shown in the remainder of the paper are also representative of those in the appendix. For each distribution and number of agents, we ran each algorithm 50 times. Figure \ref{fig:three graphs} reports the run times of SMART, BOSS, ODP-IP and POI. 
On all distributions, SMART was the fastest for all numbers of agents.  For example, after 2 seconds, with the Normal distribution for 24 agents, SMART returns optimal solutions roughly 92\% faster than BOSS, 91\% faster than ODP-IP and 54\% faster than POI, while outperforming them by multiple orders of magnitude as can be seen in Figure \ref{fig:three graphs}. The reason for this is twofold. First, when problems are hard to solve (see, for instance, the results for Exponential), the CDP part of SMART finishes before the other algorithms as it presents the best worst-case time performance. 
The second reason is that for problems where the search of a specific percentage of the solution subspaces is sufficient to find the optimal solution, the combination of GRAD and DIPS  achieves the best run time. On one hand, GRAD searches that percentage of subspaces with the best run time. On the other hand, DIPS distributes the search to further accelerate it. 
Notice that the relative contribution of each technique depends on the specific problem instance. Generally, for easier-to-solve instances, DIPS and GRAD play a more significant role in finding the optimal solution, as they target specific subspaces with the upper bound for DIPS and percentages for GRAD. As the problem becomes more difficult, CDP becomes increasingly important for searching a larger portion of the solution space, as it aims to search the entire solution space. In the worst-case scenario, CDP is the fastest technique to search the entire solution space. Hence, all algorithms have a goal and help each other achieve it as explained in Section 4.  
Additional experimental insights are described in the appendix.

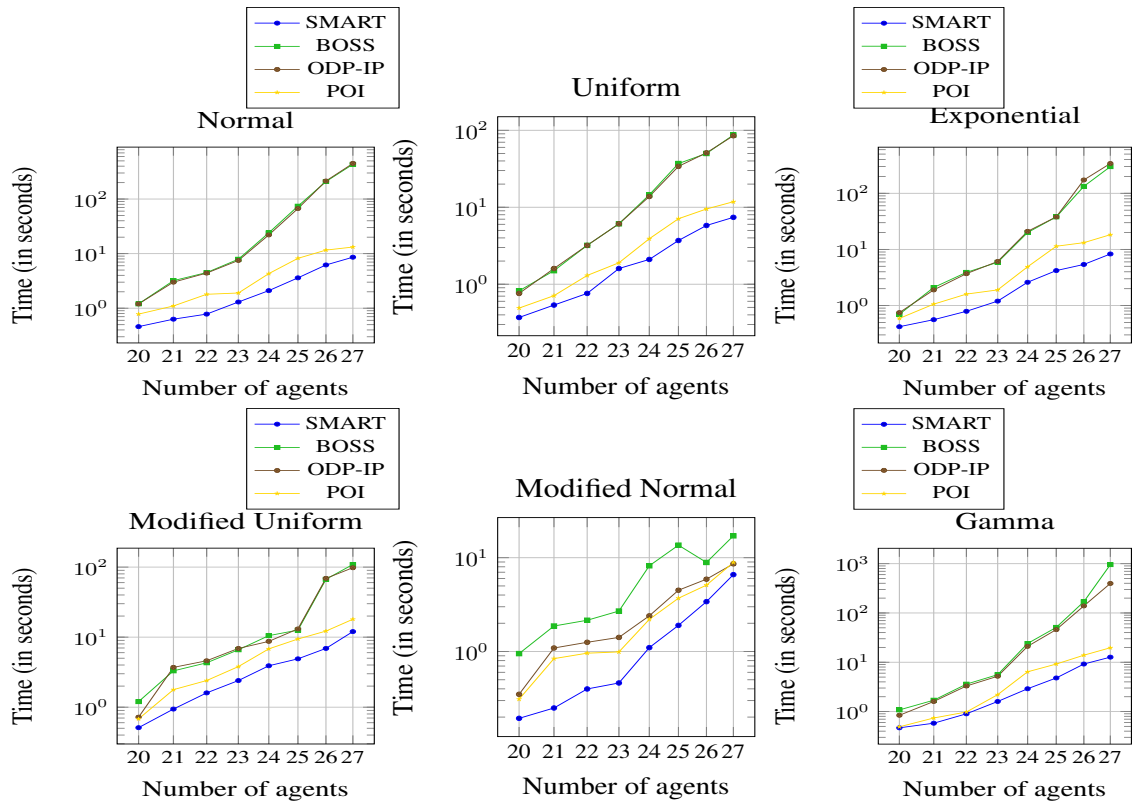
\begin{figure*}[h!]
\begin{center}
\small
\resizebox{0.28\textwidth}{5.3cm}{%
     \begin{subfigure}{0.28\textwidth}
         \centering
         \begin{tikzpicture}
	\begin{loglogaxis}[	grid= major ,
            title=Normal ,
			width=0.8\textwidth ,
            xlabel = {Number of agents} ,
			ylabel = {Time (in seconds)} ,
            width=5cm,height=5cm,
			xtick={20,21,22,23,24,25,26,27},
            xticklabels={20,21,22,23,24,25,26,27},
            ymode=log,
            ymin=0,
            label style={font=\large},
			title style={font=\Large},
			tick label style={font=\footnotesize},
            legend entries={SMART, BOSS, ODP-IP, POI},
			legend style={at={(0.8,1.715)},anchor=north}]
			\addplot+[ultra thin,mark size=1pt] coordinates {(20,0.46) (21,0.63) (22,0.78)  (23,1.3) (24,2.1) (25,3.6) (26,6.2) (27,8.6) };
			\addplot+[ultra thin,mark size=1pt,vert,mark options={fill=vert}] coordinates {(20,1.2) (21,3.2) (22,4.5) (23,7.9) (24,24) (25,74) (26,209) (27,433) };
            \addplot+[ultra thin,mark size=1pt] coordinates {(20,1.2) (21,3.0) (22,4.4) (23,7.5) (24,22) (25,67) (26,214) (27,449) };
            \addplot+[ultra thin,mark size=1pt,gold] coordinates {(20,0.78) (21,1.1) (22,1.8) (23,1.9) (24,4.3) (25,8.2) (26,11.6) (27,13.2) };
            
	\end{loglogaxis}
    \end{tikzpicture}\\
         \label{fig:three sin x}
     \end{subfigure}
     }
\resizebox{0.28\textwidth}{5.3cm}{%
     \begin{subfigure}{0.28\textwidth}
         \centering
         \begin{tikzpicture}
	\begin{loglogaxis}[	grid= major ,
            title=Uniform ,
			width=0.8\textwidth ,
            xlabel = {Number of agents} ,
			ylabel = {Time (in seconds)} ,
            width=5cm,height=5cm,
			xtick={20,21,22,23,24,25,26,27},
            xticklabels={20,21,22,23,24,25,26,27},
            ymode=log,
            label style={font=\large},
			title style={font=\Large},
			tick label style={font=\footnotesize},
            legend entries={.,.,.,.},
			legend style={at={(1.75,1.5)},anchor=north}]
			\addplot+[ultra thin,mark size=1pt] coordinates {(20,0.37) (21,0.536) (22,0.76) (23,1.6) (24,2.1) (25,3.7) (26,5.8) (27,7.4) }; 
			\addplot+[ultra thin,mark size=1pt,vert,mark options={fill=vert}] coordinates {(20,0.82) (21,1.5) (22,3.2) (23,6.1) (24,14.5) (25,37) (26,50) (27,87) };
			\addplot+[ultra thin,mark size=1pt] coordinates {(20,0.76) (21,1.6) (22,3.2) (23,6.1) (24,13.8) (25,34) (26,51) (27,85) };
			\addplot+[ultra thin,mark size=1pt,gold] coordinates {(20,0.49) (21,0.71) (22,1.3) (23,1.9) (24,3.9) (25,7.1) (26,9.5) (27,11.8) };
			
	\end{loglogaxis}
    \end{tikzpicture}\\
         \label{fig:three sin x}
     \end{subfigure}
     }
     \resizebox{0.28\textwidth}{5.3cm}{%
     \begin{subfigure}{0.28\textwidth}
         \centering
         \begin{tikzpicture}
	\begin{loglogaxis}[	grid= major ,
            title=Exponential ,
			width=0.8\textwidth ,
            xlabel = {Number of agents} ,
			ylabel = {Time (in seconds)} ,
            width=5cm,height=5cm,
			xtick={20,21,22,23,24,25,26,27},
            xticklabels={20,21,22,23,24,25,26,27},
            ymode=log,
            label style={font=\large},
			title style={font=\Large},
			tick label style={font=\footnotesize},
            legend entries={SMART, BOSS, ODP-IP, POI},
			legend style={at={(0.2,1.715)},anchor=north}]
			\addplot+[ultra thin,mark size=1pt] coordinates {(20,0.42) (21,0.56) (22,0.79) (23,1.2) (24,2.6) (25,4.2) (26,5.4) (27,8.3) }; 
			\addplot+[ultra thin,mark size=1pt,vert,mark options={fill=vert}] coordinates {(20,0.7) (21,2.1) (22,3.9) (23,5.9) (24,20) (25,38) (26,133) (27,299) };
			\addplot+[ultra thin,mark size=1pt] coordinates {(20,0.75) (21,1.9) (22,3.7)(23,6.1) (24,21) (25,38) (26,173) (27,340) };
			\addplot+[ultra thin,mark size=1pt,gold] coordinates {(20,0.59) (21,1.06) (22,1.6) (23,1.9) (24,4.9) (25,11.4) (26,13.2) (27,18.3) };
			
	\end{loglogaxis}
    \end{tikzpicture}\\
         \label{fig:three sin x}
     \end{subfigure}
}
\hfill
     \resizebox{0.28\textwidth}{5.3cm}{%
     \begin{subfigure}{0.28\textwidth}
         \centering
         \begin{tikzpicture}
	\begin{loglogaxis}[	grid= major ,
            title=Modified Uniform ,
			width=0.8\textwidth ,
            xlabel = {Number of agents} ,
			ylabel = {Time (in seconds)} ,
            width=5cm,height=5cm,
            xtick={20,21,22,23,24,25,26,27},
            xticklabels={20,21,22,23,24,25,26,27},
            ymode=log,
            label style={font=\large},
			title style={font=\Large},
			tick label style={font=\footnotesize},
            legend entries={SMART, BOSS, ODP-IP, POI},
			legend style={at={(0.8,1.715)},anchor=north}]
			\addplot+[ultra thin,mark size=1pt] coordinates {(20,0.51) (21,0.94) (22,1.6) (23,2.4) (24,3.9) (25,4.9) (26,6.9) (27,12) };
			\addplot+[ultra thin,mark size=1pt,vert,mark options={fill=vert}] coordinates {(20,1.2) (21,3.3) (22,4.3) (23,6.6) (24,10.5) (25,12.6) (26,67) (27,109) };
			\addplot+[ultra thin,mark size=1pt] coordinates {(20,0.72) (21,3.7) (22,4.6) (23,6.9) (24,8.7) (25,13.1) (26,69) (27,98) };
			\addplot+[ultra thin,mark size=1pt,gold] coordinates {(20,0.68) (21,1.78) (22,2.4) (23,3.8) (24,6.8) (25,9.4) (26,12.2) (27,18) };

	\end{loglogaxis}
    \end{tikzpicture}\\
         \label{fig:three sin x}
     \end{subfigure}
     }
     \resizebox{0.28\textwidth}{5.3cm}{%
     \begin{subfigure}{0.28\textwidth}
         \centering
         \begin{tikzpicture}
	\begin{loglogaxis}[	grid= major ,
            title=Modified Normal ,
			width=0.8\textwidth ,
            xlabel = {Number of agents} ,
			ylabel = {Time (in seconds)} ,
            width=5cm,height=5cm,
			xtick={20,21,22,23,24,25,26,27},
            xticklabels={20,21,22,23,24,25,26,27},
            ymode=log,
            label style={font=\large},
			title style={font=\Large},
			tick label style={font=\footnotesize},
            legend entries={.,.,.},
			legend style={at={(1.75,1.5)},anchor=north}]
			\addplot+[ultra thin,mark size=1pt] coordinates {(20,0.194) (21,0.25) (22,0.399) (23,0.462) (24,1.1) (25,1.9) (26,3.4) (27,6.6) }; 
			\addplot+[ultra thin,mark size=1pt,vert,mark options={fill=vert}] coordinates {(20,0.95) (21,1.87) (22,2.154) (23,2.7) (24,8.2) (25,13.6) (26,8.9) (27,17.1) };
			\addplot+[ultra thin,mark size=1pt] coordinates {(20,0.35) (21,1.09) (22,1.254) (23,1.414) (24,2.4) (25,4.5) (26,5.9) (27,8.6) };
			\addplot+[ultra thin,mark size=1pt,gold] coordinates {(20,0.31) (21,0.84) (22,0.96) (23,0.99) (24,2.2) (25,3.7) (26,5.1) (27,8.9) };
			
	\end{loglogaxis}
    \end{tikzpicture}\\
         \label{fig:three sin x}
     \end{subfigure}
     }
     \resizebox{0.28\textwidth}{5.3cm}{%
     \begin{subfigure}{0.28\textwidth}
         \centering
         \begin{tikzpicture}
	\begin{loglogaxis}[	grid= major ,
            title=Gamma ,
			width=0.8\textwidth ,
            xlabel = {Number of agents} ,
			ylabel = {Time (in seconds)} ,
            width=5cm,height=5cm,
			xtick={20,21,22,23,24,25,26,27},
            xticklabels={20,21,22,23,24,25,26,27},
            ymode=log,
            label style={font=\large},
			title style={font=\Large},
			tick label style={font=\footnotesize},
            legend entries={SMART, BOSS, ODP-IP, POI},
			legend style={at={(0.2,1.715)},anchor=north}]
			\addplot+[ultra thin,mark size=1pt] coordinates {(20,0.468) (21,0.58) (22,0.9) (23,1.6)  (24,2.9) (25,4.8) (26,9.2) (27,12.7) }; 
			\addplot+[ultra thin,mark size=1pt,vert,mark options={fill=vert}] coordinates {(20,1.08) (21,1.7) (22,3.6) (23,5.6)  (24,24) (25,51) (26,169) (27,957) };
			\addplot+[ultra thin,mark size=1pt] coordinates {(20,0.84) (21,1.6) (22,3.3) (23,5.2) (24,21) (25,46) (26,139) (27,397) };
			\addplot+[ultra thin,mark size=1pt,gold] coordinates {(20,0.49) (21,0.74) (22,0.98) (23,2.2) (24,6.4) (25,9.2) (26,13.9) (27,19.7) };
			
	\end{loglogaxis}
    \end{tikzpicture}\\
         \label{fig:three sin x}
     \end{subfigure}
}

             \caption{\small Run time of SMART, BOSS, ODP-IP, and POI.}
        \label{fig:three graphs}
\normalsize
\end{center}
\end{figure*}

We also report the empirical performance of CDP, which, as discussed earlier in this paper, determines the worst-case run time of SMART. We compared CDP to the dynamic programming algorithm IDP~\cite{rahwan2008improved} used by prior hybrid algorithms, and to the fastest dynamic programming algorithm to date, ODP~\cite{michalak2016hybrid}. Notice that ODP, which stands for Optimal DP, is optimal in the sense that it evaluates a minimum number of coalitions to find the optimal solution, not in the sense of run time, meaning that the number of evaluated coalitions is optimal. This does not mean that the resulting run time is optimal. We also compared CDP to a parallel version of IDP  (P-IDP). P-IDP uses the same technique presented in~\cite{cruz2017coalition} and uses two processes to evaluate the coalitions. The evaluation of the coalitions of the same size can be distributed because the coalitions of the same size are independent of each other. Hence, for every coalition size 
$s$, each process of P-IDP evaluates half of the coalitions of size $s$. 
\begin{table}
\begin{center}
\renewcommand{\arraystretch}{1.0}
\begin{tabular}{|c|c|c|c|c|c|}
\hline
\textbf{Number of Agents} & \multicolumn{4}{c|}{\textbf{Execution Time}} \\
\cline{2-5}
 &  \textbf{CDP} & \textbf{IDP} & \textbf{ODP} & \textbf{P-IDP} \\
  &  ($t_{1}$) & ($t_{2}$) &  &  \\
\hline
20 & $1.7$ & $3.7$ & $2.2$ & $2.0$ \\
\hline
21 & $7.4$ & $15.9$ & $10.3$ & $7.9$ \\
\hline
22 & $12.3$ & $24.7$ & $19.8$ & $15.7$ \\
\hline
23 & $57$ & $131$ & $79$ & $69$ \\
\hline
24 & $205$ & $507$ & $382$ & $257$ \\
\hline
25 & $427$ & $887$ & $675$ & $593$ \\
\hline
26 & $1781$ & $3659$ & $2846$ & $2267$ \\
\hline
27 & $3390$ & $7078$ & $5508$ & $5196$ \\
\hline
\end{tabular}
\end{center}
\caption{\small Time in seconds of CDP, IDP, ODP, and P-IDP. The time  gain is shown in the appendix.}\label{tabresult}
\end{table}
Table~\ref{tabresult} shows the results. The run time of these algorithms depends only on the number of agents. As can be seen, CDP outperforms IDP by at least $50\%$. Moreover, CDP is also faster than ODP and P-IDP (See the appendix for the time difference between CDP and P-IDP). This experimentally confirms that SMART offers the best worst-case run time. This superior speed appears to translate into the superior practical performance of the SMART algorithm as well. With the most difficult distributions, such as Gamma (see Figure \ref{fig:three graphs}), the SMART algorithm is significantly faster than the other algorithms. 

\section{Conclusion}

In this paper, we developed an optimal algorithm, SMART, for the coalition structure generation problem. Our method contributes and combines  a number of ideas and techniques. First, we introduced several results concerning the choice of coalitions to evaluate. We used those results to build offline phases to optimize the choice of coalitions to evaluate. Second, we developed three techniques that have different pros. Two of them use the results of the offline phases. The third one uses branch-and-bound and integer partition graph search to explore the solution space. 
Finally, we combined these techniques by showing how they can assist one another during the search process. 
Experiments showed that SMART is faster than the fastest prior algorithms on all of the instance distributions for all numbers of agents.

\section*{Acknowledgments}

Tuomas Sandholm's research is supported by the Vannevar Bush Faculty Fellowship ONR N00014-23-1-2876, National Science Foundation grant RI-2312342, ARO award W911NF2210266, and NIH award A240108S001.

\bibliographystyle{named}
\bibliography{ijcai24}

\newpage
\appendix

In this appendix, we provide supplementary material. The extensive related work section has been included in the appendix to offer a thorough exploration of existing algorithms.

\section{Related Work}

Efficiently solving the coalition structure generation problem is computationally expensive when using a naive approach that involves enumerating all possible coalition structures~\cite{sandholm1999coalition}. To overcome this challenge, various representations of the search space have been proposed to reduce the time required to generate optimal coalition structures. 

\subsection{Coalition Structure Graph}

The coalition structure graph, first introduced by ~\cite{sandholm1999coalition}, is a way to represent the search space as a graph composed of nodes representing the coalition structures. For a given set of $n$ agents, these nodes are organized into $n$ levels, where each level consists of nodes representing coalition structures that contain exactly $i$ coalitions ($i \in \{1,..,n\}$). Each edge of this graph connects two nodes belonging to two consecutive levels, such that each coalition structure at level $i$ can be obtained by dividing a coalition from a coalition structure at level $i - 1$ into two coalitions. The graph of the coalition structures of four agents is shown in Figure \ref{CSG1}.

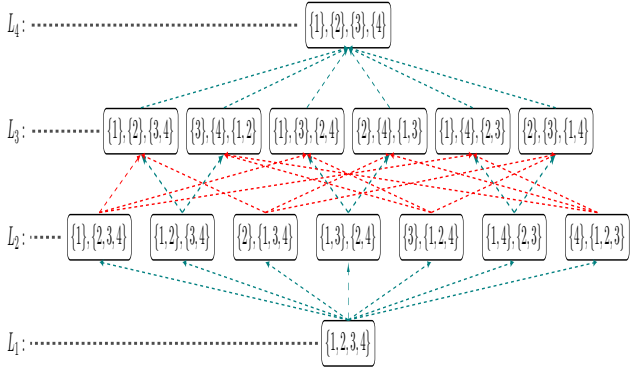
\begin{figure}[h!]
\centering
\normalsize
\resizebox{0.97\columnwidth}{4.85cm}{%
\begin{tikzpicture}
\large

\node[] (a)at(-4.5,-4.5){$L_{1}: $};
\draw[>=latex,black!70,dashed,very thick] (6.4,-4.5) -- (-4.0,-4.5);
\node[] (a)at(-4.5,-3){$L_{2}: $};
\draw[>=latex,black!70,dashed,very thick] (-2.9,-3) -- (-4.0,-3);
\node[] (a)at(-4.5,-1.5){$L_{3}: $};
\draw[>=latex,black!70,dashed,very thick] (-1.5,-1.5) -- (-4.0,-1.5);
\node[] (a)at(-4.5,0){$L_{4}: $};
\draw[>=latex,black!70,dashed,very thick] (5.65,0) -- (-4.0,0);

\node[draw,rectangle,rounded corners=3pt,black] (a)at(7.5,0){\color{black}$\{1\},\{2\},\{3\},\{4\}$};
\node[draw,rectangle,rounded corners=3pt,black] (b)at(0,-1.5){$\{1\},\{2\},\{3,4\}$};
\node[draw,rectangle,rounded corners=3pt,black] (c)at(3,-1.5){$\{3\},\{4\},\{1,2\}$};
\node[draw,rectangle,rounded corners=3pt,black] (d)at(6,-1.5){$\{1\},\{3\},\{2,4\}$};
\node[draw,rectangle,rounded corners=3pt,black] (e)at(9,-1.5){$\{2\},\{4\},\{1,3\}$};
\node[draw,rectangle,rounded corners=3pt,black] (f)at(12,-1.5){\color{black}$\{1\},\{4\},\{2,3\}$};
\node[draw,rectangle,rounded corners=3pt,black] (g)at(15,-1.5){\color{black}$\{2\},\{3\},\{1,4\}$};

\node[draw,rectangle,rounded corners=3pt,black] (h)at(-1.5,-3){\color{black}$\{1\},\{2,3,4\}$};
\node[draw,rectangle,rounded corners=3pt,black] (i)at(1.5,-3){\color{black}$\{1,2\},\{3,4\}$};
\node[draw,rectangle,rounded corners=3pt,black] (j)at(4.5,-3){\color{black}$\{2\},\{1,3,4\}$};
\node[draw,rectangle,rounded corners=3pt,black] (k)at(7.5,-3){\color{black}$\{1,3\},\{2,4\}$};
\node[draw,rectangle,rounded corners=3pt,black] (l)at(10.5,-3){\color{black}$\{3\},\{1,2,4\}$};
\node[draw,rectangle,rounded corners=3pt,black] (m)at(13.5,-3){\color{black}$\{1,4\},\{2,3\}$};
\node[draw,rectangle,rounded corners=3pt,black] (n)at(16.5,-3){\color{black}$\{4\},\{1,2,3\}$};
\node[draw,rectangle,rounded corners=3pt,black] (o)at(7.5,-4.5){\color{black}$\{1,2,3,4\}$};

\draw[->,>=latex,teal,dashed] (0,-1.17) -- (7.5,-0.32);
\draw[->,>=latex,teal,dashed] (3,-1.17) -- (7.5,-0.32);
\draw[->,>=latex,teal,dashed] (6,-1.17) -- (7.5,-0.32);
\draw[->,>=latex,teal,dashed] (9,-1.17) -- (7.5,-0.32);
\draw[->,>=latex,teal,dashed] (12,-1.17) -- (7.5,-0.32);
\draw[->,>=latex,teal,dashed] (15,-1.17) -- (7.5,-0.32);

\draw[->,>=latex,red,dashed] (-1.5,-2.65) -- (0,-1.83);
\draw[->,>=latex,red,dashed] (-1.5,-2.65) -- (6,-1.83);
\draw[->,>=latex,red,dashed] (-1.5,-2.65) -- (12,-1.83);
\draw[->,>=latex,teal,thick,dashed] (1.5,-2.65) -- (0,-1.83);
\draw[->,>=latex,teal,thick,dashed] (1.5,-2.65) -- (3,-1.83);
\draw[->,>=latex,red,dashed] (4.5,-2.65) -- (0,-1.83);
\draw[->,>=latex,red,dashed] (4.5,-2.65) -- (9,-1.83);
\draw[->,>=latex,red,dashed] (4.5,-2.65) -- (15,-1.83);
\draw[->,>=latex,teal,thick,dashed] (7.5,-2.65) -- (6,-1.83);
\draw[->,>=latex,teal,thick,dashed] (7.5,-2.65) -- (9,-1.83);
\draw[->,>=latex,red,dashed] (10.5,-2.65) -- (3,-1.83);
\draw[->,>=latex,red,dashed] (10.5,-2.65) -- (6,-1.83);
\draw[->,>=latex,red,dashed] (10.5,-2.65) -- (15,-1.83);
\draw[->,>=latex,teal,thick,dashed] (13.5,-2.65) -- (12,-1.83);
\draw[->,>=latex,teal,thick,dashed] (13.5,-2.65) -- (15,-1.83);
\draw[->,>=latex,red,dashed] (16.5,-2.65) -- (3,-1.83);
\draw[->,>=latex,red,dashed] (16.5,-2.65) -- (9,-1.83);
\draw[->,>=latex,red,dashed] (16.5,-2.65) -- (12,-1.83);

\draw[->,>=latex,teal,dashed] (7.5,-4.17) -- (-1.5,-3.35);
\draw[->,>=latex,teal,dashed] (7.5,-4.17) -- (1.5,-3.35);
\draw[->,>=latex,teal,dashed] (7.5,-4.17) -- (4.5,-3.35);
\draw[->,>=latex,teal,dashed] (7.5,-4.17) -- (7.5,-3.35);
\draw[->,>=latex,teal,dashed] (7.5,-4.17) -- (10.5,-3.35);
\draw[->,>=latex,teal,dashed] (7.5,-4.17) -- (13.5,-3.35);
\draw[->,>=latex,teal,dashed] (7.5,-4.17) -- (16.5,-3.35);

            \end{tikzpicture}%
        }
\caption{Coalition structure graph of 4 agents.} \label{CSG1}
\end{figure}

\subsection{Integer Partition Graph}

The \textit{integer partition graph}~\cite{rahwan2009anytime} is a representation of all possible coalition structures. Given $n$ agents, with this representation, each integer partition of $n$ is represented by a node, where two adjacent nodes are connected if and only if the integer partition in level $i$ can be reached from the one in level $i-1$ by splitting only one integer. For instance, for $n=4$, the set of partitions is: $\{[4],[1,3],[2,2],[1,1,2],[1,1,1,1]\}$. Each partition of $n$ sums to $n$. Figure \ref{IPG} shows a four-agent example of the integer partition graph. 

Each level 
$i \in \{1,2,..,n\}$ in this graph contains nodes representing integer partitions of $n$ containing $i$ parts. For instance, level 2 contains nodes where integer partitions of $n$ have two parts. In the graph, each partition $\mathcal{P}$ represents a subspace of solutions that contains a set of coalition structures in which the sizes of the coalitions match the parts of $\mathcal{P}$. For example, the node [1,1,2] consists of all coalition structures that contain two coalitions of size 1 and one coalition of size 2.~\cite{9643288} presented another layer above this graph using code vectors of coalition structures and~\cite{10.1007/978-3-030-69322-0_2} proposed a generalization of this graph to be able to capture task allocation. 

\begin{figure}[h!]
\centering
\resizebox{1.0\columnwidth}{5.3cm}{%
\begin{tikzpicture}
\large
\node[draw,rectangle,rounded corners=3pt] (a)at(7.5,0){$[1,1,1,1]$};

\node[draw,rectangle,rounded corners=3pt] (a)at(7.5,-2){$[1,1,2]$};

\node[draw,rectangle,rounded corners=3pt] (a)at(5.5,-4){$[1,3]$};

\node[draw,rectangle,rounded corners=3pt] (a)at(9.5,-4){$[2,2]$};

\node[draw,rectangle,rounded corners=3pt] (a)at(7.5,-6){$[4]$};

\node[] (a)at(-1.5,-6){$L_{1}: $};
\draw[>=latex,black!70,dashed,very thick] (7.0,-6) -- (-1.0,-6);
\node[] (a)at(-1.5,-4){$L_{2}: $};
\draw[>=latex,black!70,dashed,very thick] (0.1,-4) -- (-1.0,-4);
\node[] (a)at(-1.5,-2){$L_{3}: $};
\draw[>=latex,black!70,dashed,very thick] (6.5,-2) -- (-1.0,-2);
\node[] (a)at(-1.5,0){$L_{4}: $};
\draw[>=latex,black!70,dashed,very thick] (6.25,0) -- (-1.0,0);

\node[] (a)at(9.5,-6){$\Pi_{[4]}: $};
\node[] (a)at(11.3,-6){$\{\{a_{1},a_{2},a_{3},a_{4}\}\}$};

\node[] (a)at(10.8,-4){$\Pi_{[2,2]}: $};
\node[] (a)at(13,-4){$\{\{a_{1},a_{2}\},\{a_{3},a_{4}\}\}$};
\node[] (a)at(13,-4.5){$\{\{a_{1},a_{3}\},\{a_{2},a_{4}\}\}$};
\node[] (a)at(13,-5){$\{\{a_{1},a_{4}\},\{a_{2},a_{3}\}\}$};

\node[] (a)at(0.8,-4){$\Pi_{[1,3]}: $};
\node[] (a)at(3,-4){$\{\{a_{1}\},\{a_{2},a_{3},a_{4}\}\}$};
\node[] (a)at(3,-4.5){$\{\{a_{2}\},\{a_{1},a_{3},a_{4}\}\}$};
\node[] (a)at(3,-5){$\{\{a_{3}\},\{a_{1},a_{2},a_{4}\}\}$};
\node[] (a)at(3,-5.5){$\{\{a_{4}\},\{a_{1},a_{2},a_{3}\}\}$};

\node[] (a)at(9.0,-2){$\Pi_{[1,1,2]}: $};
\node[] (a)at(11.5,-2){$\{\{a_{1}\},\{a_{2}\},\{a_{3},a_{4}\}\}$};
\node[] (a)at(11.5,-2.5){$\{\{a_{1}\},\{a_{3}\},\{a_{2},a_{4}\}\}$};
\node[] (a)at(11.5,-3){$\{\{a_{1}\},\{a_{4}\},\{a_{2},a_{3}\}\}$};
\node[] (a)at(15.3,-2){$\{\{a_{2}\},\{a_{3}\},\{a_{1},a_{4}\}\}$};
\node[] (a)at(15.3,-2.5){$\{\{a_{2}\},\{a_{4}\},\{a_{1},a_{3}\}\}$};
\node[] (a)at(15.3,-3){$\{\{a_{3}\},\{a_{4}\},\{a_{1},a_{2}\}\}$};

\node[] (a)at(13.4,-2){$,$};
\node[] (a)at(13.4,-2.5){$,$};
\node[] (a)at(13.4,-3){$,$};

\node[] (a)at(9.5,0){$\Pi_{[1,1,1,1]}: $};
\node[] (a)at(12.3,0){{$\{\{a_{1}\},\{a_{2}\},\{a_{3}\},\{a_{4}\}\}$}};

\node[] (a)at(9.3,-5.3){\small$4=2+2$};
\node[] (a)at(9.3,-4.9){\small Split $4$};
\node[] (a)at(5.7,-5.3){\small$4=1+3$};
\node[] (a)at(5.7,-4.9){\small Split $4$};

\node[] (a)at(5.2,-3.3){\small$3=1+2$};
\node[] (a)at(5.2,-2.9){\small Split $3$};
\node[] (a)at(7.95,-3.5){\small$2=1+1$};
\node[] (a)at(7.95,-3.1){\small Split $2$};

\node[] (a)at(6.7,-1.2){\small$2=1+1$};
\node[] (a)at(6.7,-0.8){\small Split $2$};

\draw[->,>=latex,teal,very thick] (7.5,-5.7) -- (5.5,-4.3);
\draw[->,>=latex,teal,very thick] (7.5,-5.7) -- (9.5,-4.3);
\draw[->,>=latex,red,dashed,very thick] (5.5,-3.7) -- (7.4,-2.3);
\draw[->,>=latex,teal,very thick] (9.5,-3.7) -- (7.5,-2.3);
\draw[->,>=latex,teal,thick] (7.5,-1.7) -- (7.5,-0.3);

\end{tikzpicture}%
}
\caption{Four-agent integer partition graph. It has four levels: $L_{1}$ to $L_{4}$, numbered from bottom to top.} \label{IPG}
\end{figure}
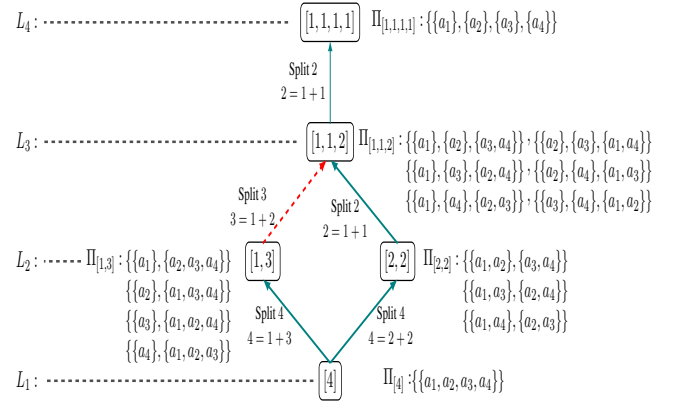

In the following subsections, we will explore both exact and approximate methods for solving the coalition structure generation problem. This problem is notoriously difficult to solve, and there is ongoing research on developing more efficient algorithms. For a more in-depth analysis of these approaches, we refer interested readers to~\cite{rahwan2015coalition}.

\subsection{Exact approaches}

Several approaches have been proposed in prior studies for optimally solving the coalition structure generation problem, including dynamic programming and anytime algorithms, which have achieved some success~\cite{rahwan2015coalition}. Yeh~\cite{yeh1986dynamic} was the first to propose a dynamic programming (DP) algorithm to solve the complete set partitioning problem, which was later rediscovered by Rothkopf et al.~\cite{rothkopf1998computationally} for solving the winner determination problem in combinatorial auctions. This algorithm guarantees finding an optimal solution, but it must be run to completion. An improved dynamic programming algorithm was proposed by Rahwan et al.~\cite{rahwan2008improved}, which showed that some calculations in the algorithm by Yeh~\cite{yeh1986dynamic} were unnecessary for finding the optimal solution. Michalak et al.~\cite{michalak2016hybrid} proposed an optimal dynamic programming algorithm (ODP) that avoids many evaluation operations performed by DP, making it evaluate an optimal set of coalitions to ensure an optimal solution. ODP runs in $O(3^n)$ and provides no intermediate solutions. Bjorklund et al.~\cite{bjorklund2009set} proposed an inclusion-exclusion-based dynamic programming algorithm that theoretically runs in $O(2^n)$ time but was shown to run in $6^n$ time when tested in practice~\cite{michalak2016hybrid}. Changder and Aknine~\cite{changder2019effective} developed another dynamic programming algorithm by proving that evaluating the coalitions of even sizes is sufficient for finding an optimal solution.

To allow premature termination while providing guaranteed bounds on the quality of the solution found at each moment during the coalition structure search process, many works have developed anytime algorithms. Sandholm et al.~\cite{sandholm1999coalition} proposed the first anytime algorithm based on a coalition structure graph (see Figure 1), where searching the lowest two levels of the coalition structure graph guarantees a certain quality of the solution found. Similarly, Dang et al.~\cite{dang2004generating} used the same coalition structure graph, and their algorithm was shown to empirically generate tighter quality guarantees than Sandholm et al.'s algorithm. Rahwan et al.~\cite{rahwan2009anytime} developed an anytime algorithm based on a new representation of the search space, the integer partition graph (see Figure \ref{IPG}), called IP, which groups coalition structures according to the number and size of all coalitions they contain and uses a branch and bound technique. This algorithm finds higher quality solutions faster than previous anytime algorithms~\cite{rahwan2009anytime}. Michalak et al.~\cite{michalak2010distributed} developed a decentralized version of IP. However, the worst-case running time of current anytime algorithms is $O(n^n)$~\cite{rahwan2015coalition}.

Several algorithms have been developed that combine techniques from the dynamic programming approach and the anytime approach. Michalak et al.~\cite{michalak2016hybrid} proposed the first such algorithm, namely, ODP-IP, which combines IDP and IP. To date, the hybrid algorithms ODP-IP, ODSS~\cite{changder2020odss}, and BOSS~\cite{Changder_Aknine_Ramchurn_Dutta_2021} are the fastest exact algorithms for the CSG problem and are efficient for many instances when the time required to produce the optimal solution is shorter than the available time for running the algorithm. 

\subsection{Heuristic Algorithms}

There have been several heuristic algorithms proposed for the coalition structure generation problem (CSG) that prioritize speed over finding an optimal solution. When the number of agents increases and the problem becomes too difficult, these algorithms are often the only practical option~\cite{rahwan2015coalition}. One of the earliest heuristic algorithms for CSG focused on task allocation and was proposed by Shehory and Kraus~\cite{shehory1993coalition,shehory1995task,shehory1996formation,shehory1998methods}. These algorithms restrict the size of evaluated coalitions to a certain number of agents.~\cite{sen2000searching} introduced a genetic algorithm for CSG, where coalition structures represent the population and are selected based on the outcome of the evaluation. Other heuristic approaches have been studied for the CSG problem, such as simulated annealing~\cite{keinanen2009simulated}, greedy-based methods~\cite{di2010coalition}, tabu search~\cite{hussin2016heuristic}, and code permutation search~\cite{Taguelmimt_Aknine_Boukredera_Changder_2021}.

\subsection{Scalable Solutions}

The coalition structure generation (CSG) problem is notoriously challenging to solve at scale, with very few scalable solutions available. One such solution, proposed by~\cite{wu2020monte}, is based on Monte-Carlo Tree Search, which finds solutions by sampling the coalition structure graph and partially expanding a search tree that corresponds to the explored partial search space. Although this approach can provide optimal solutions, it requires generating the entire search tree by adding all nodes. Another solution proposed by~\cite{farinelli2013c} is a hierarchical clustering approach that builds a high-quality coalition structure by merging coalitions based on a similarity criterion, and a similar approach based on GRASP was proposed in~\cite{di2010coalition}. However, these methods are limited in that they only consider the value of the best merge of two coalitions, with additional operations of splitting a coalition into two and exchanging a pair of agents between two coalitions. In addition, the search algorithms FACS and PICS proposed in~\cite{9643288,10098066} generate coalition structures based on code permutations applied to selected initial vectors of a different search space representation. To the best of our knowledge, the PICS~\cite{10098066} and CSG-UCT~\cite{wu2020monte} are the best performing of the prior algorithms. Although some success has been achieved in solving large-scale CSG problems with these techniques, scalability remains a challenging issue, and PICS currently represents the most promising technique for addressing this problem.

\section{Reducing Memory Requirements}

To save memory, we store the $V_{t}$ table only once and use it for all processes of SMART. Observe that for each CDP and GRAD process, we need to compute the value $V_{t}(C)$ for all considered coalitions in the set of sizes during each execution of the algorithm. However, given that the processes of the CDP and GRAD algorithms gradually store the maximum values of the coalitions, working with the same $V_{t}$ table does not interfere with this working. The underlying reason is that for each process of CDP and GRAD, the $V_{t}$ table stores at least the needed values. In case the value of a coalition is higher than the one computed, it only enables the CDP or GRAD process to search for more coalition structures and never less. 
Furthermore, in terms of the waiting caused by locking the shared memory location when updating the value of a coalition, this has minimal impact because each process has a different set of coalitions and hence, requires access to different locations



\section{Details on the SetTime and GeneratedSubspaces functions}

GeneratedSubspaces and SetTime are simple procedures used in SSD and SOFT. GeneratedSubspaces starts with the size $n$ and follows all the splittings of sizes that belong to the set of sizes evaluated and returns all the nodes connected to the bottom node. 

SetTime is a simple function. It sums the times associated with the sizes that belong to the size set.



Algorithm \ref{partition} shows the pseudocode of the Partition function.

{\SetAlgoNoLine%
\begin{algorithm}[h!]
\KwIn{The partition table $\mathcal{P}_{t}$. A set of agents $A$.}
\KwOut{An optimal partition of $A$.}
\DontPrintSemicolon

$\mathcal{CS}^{*} \leftarrow \{A\}$\;
\ForEach{$C \subseteq \mathcal{CS}^{*}$}{
\If($\triangleright$ if $\{C\}$ is not an optimal partition of $C$){$\mathcal{P}_{t}(C) \neq \{C\}$}{
$\mathcal{CS}^{*} \leftarrow (\mathcal{CS}^{*} \ \{C\}) \cup \mathcal{P}_{t}(C)$\;
Go to line 2 and start with the new $\mathcal{CS}^{*}$\;
}
}
Return $\mathcal{CS}^{*}$\;
 \caption{Partition} \label{partition}
\end{algorithm}}%


\section{Proofs}


\noindent \textbf{Time complexity of the SOFT algorithm.}\: 
The time complexity of SSD is $\mathcal{O}(2^{2n} \times n^{2} \times \frac{e^{\pi\sqrt{\frac{2n}{3}}}}{n})$, which corresponds to the cost of pairing each set with the other sets. In SOFT, we evaluate the sets without pairing them. This is equivalent in SSD to pairing only one set with the others. Thus, the time complexity of SOFT is $\mathcal{O}(2^{n} \times n^{2} \times \frac{e^{\pi\sqrt{\frac{2n}{3}}}}{n})$.

\section{Additional Figures}

Figure \ref{dips} shows an illustration of how DIPS proceeds on an example with ten agents and a three-part subspace $[1,2,7]$.

\begin{figure}[h!]
\small
\centering
\resizebox{0.9\columnwidth}{6.5cm}{%
\begin{tikzpicture}
\large
\node[draw,rectangle,fill=black!25!black!30,rounded corners=5pt, minimum height=1cm, minimum width=1.7cm, text width=1.6cm, text centered] (a)at(1.5,-0.5){Coalitions of size \textbf{1}};

\node[draw,rectangle,fill=LightCyan2,rounded corners=3pt, minimum height=4.0cm, minimum width=1.8cm, text width=0.3cm, text centered] (a)at(1.5,-3.05){};

\node[draw,rectangle,fill=black!25!black!30,rounded corners=5pt, minimum height=1cm, minimum width=1.7cm, text width=1.6cm, text centered] (a)at(7.5,-0.5){Coalitions of size \textbf{2}};

\node[draw,rectangle,rounded corners=3pt, minimum height=5.5cm, minimum width=1.8cm, text width=0.3cm, text centered] (a)at(7.5,-3.8){};

\node[draw,rectangle,fill=black!25!black!30,rounded corners=5pt, minimum height=1cm, minimum width=1.7cm, text width=1.6cm, text centered] (a)at(13.5,-0.5){Coalitions of size \textbf{7}};

\node[draw,rectangle,rounded corners=3pt, minimum height=6.5cm, minimum width=1.8cm, text width=0.3cm, text centered] (a)at(13.5,-4.3){};

\Large
\node[] (a)at(1.5,-1.5){$.$};
\node[] (a)at(1.5,-1.7){$.$};
\node[] (a)at(1.5,-1.9){$.$};
\node[] (a)at(1.5,-2.6){$\mathcal{C}_{x}$};
\node[] (a)at(1.5,-3.5){$\mathcal{C}_{y}$};
\node[] (a)at(1.5,-4.2){$.$};
\node[] (a)at(1.5,-4.4){$.$};
\node[] (a)at(1.5,-4.6){$.$};
\Large

\draw[->,>=latex,black, line width=0.6pt] (2,-2.6) -- (4.5,-2.6);
\draw[-,>=latex,red,very thick, line width=4pt] (4.5,-2.1) -- (5.2,-3.1);
\draw[-,>=latex,red,very thick, line width=4pt] (4.5,-3.1) -- (5.2,-2.1);

\draw[->,>=latex,red,dashed, thick, line width=1pt] (4.15,-2.5) -- (4.15,0);
\Large
\node[] (a)at(4.15,0.4){$v(\mathcal{C}_{x}) + Max_{2} + Max_{7} < v(\mathcal{CS}^{+})$};

\draw[->,>=latex,linecolor,dashed, thick, line width=1pt] (4.5,-4.2) -- (4.5,-6.3);
\Large
\node[] (a)at(4.5,-6.8){$v(\mathcal{C}_{x}) + Max_{2} + Max_{7} \geq v(\mathcal{CS}^{+})$};

\Large
\node[] (a)at(7.5,-1.5){$.$};
\node[] (a)at(7.5,-1.9){$.$};
\node[] (a)at(7.5,-2.3){$.$};
\node[] (a)at(7.5,-2.7){$.$};
\node[] (a)at(7.5,-3.5){$\mathcal{C}_{i}$};
\node[] (a)at(7.5,-4.4){$\mathcal{C}_{j}$};
\node[] (a)at(7.5,-5){$.$};
\node[] (a)at(7.5,-5.4){$.$};
\node[] (a)at(7.5,-5.8){$.$};

\scriptsize

\draw[->,>=latex,black, line width=0.6pt] (2,-3.5) -- (7.1,-3.5);
\draw[->,>=latex,black, line width=0.6pt] (2,-3.5) -- (7.1,-4.4);

\Large
\node[] (a)at(13.5,-1.5){$.$};
\node[] (a)at(13.5,-2){$.$};
\node[] (a)at(13.5,-2.5){$.$};
\node[] (a)at(13.5,-3){$.$};
\node[] (a)at(13.5,-3.5){$.$};
\node[] (a)at(13.5,-4.4){$\mathcal{C}_{p}$};
\node[] (a)at(13.5,-5.3){$\mathcal{C}_{q}$};
\node[] (a)at(13.5,-5.9){$.$};
\node[] (a)at(13.5,-6.4){$.$};
\node[] (a)at(13.5,-6.9){$.$};

\Large

\draw[->,>=latex,black, line width=0.6pt] (8,-3.5) -- (10.5,-3.5);
\draw[-,>=latex,red,very thick, line width=4pt] (10.5,-3) -- (11.2,-4);
\draw[-,>=latex,red,very thick, line width=4pt] (10.5,-4) -- (11.2,-3);

\draw[->,>=latex,black, line width=0.6pt] (8,-4.4) -- (13.1,-4.4);
\draw[->,>=latex,black, line width=0.6pt] (8,-4.4) -- (13.1,-5.3);

\draw[->,>=latex,red,dashed, thick, line width=1pt] (10.15,-3.4) -- (10.15,0);
\Large
\node[] (a)at(10.65,0.4){$v(\mathcal{C}_{y}) + v(\mathcal{C}_{i}) + Max_{7} < v(\mathcal{CS}^{+})$};

\draw[->,>=latex,linecolor,dashed, thick, line width=1pt] (10.5,-5.1) -- (10.5,-7.6);
\Large
\node[] (a)at(9.8,-8.1){$v(\mathcal{C}_{y}) + v(\mathcal{C}_{j}) + Max_{7} \geq v(\mathcal{CS}^{+})$};

\end{tikzpicture}%
}
\caption{Illustration of the branch-and bound technique when searching the node $[1,2,7]$. In this example, DIPS constructs 10 search trees. The roots of these trees are the singleton coalitions (because the first part of the node is of size 1), as shown in the blue rectangle. The algorithm recognizes that the coalition structures containing the coalition $\mathcal{C}_{x}$ (because $v(\mathcal{C}_{x}) + Max_{2} + Max_{7} < v(\mathcal{CS}^{+})$) or both the coalitions $\mathcal{C}_{y}$ and $\mathcal{C}_{i}$ (because $v(\mathcal{C}_{y}) + v(\mathcal{C}_{i}) + Max_{7} < v(\mathcal{CS}^{+})$), cannot be optimal. Thus, DIPS does not search further in the tree. Here, $Max_{i}$ is the maximum value a coalition of size $i$ can take and $\mathcal{CS}^{+}$ is the last best solution found.} \label{dips}
\end{figure}
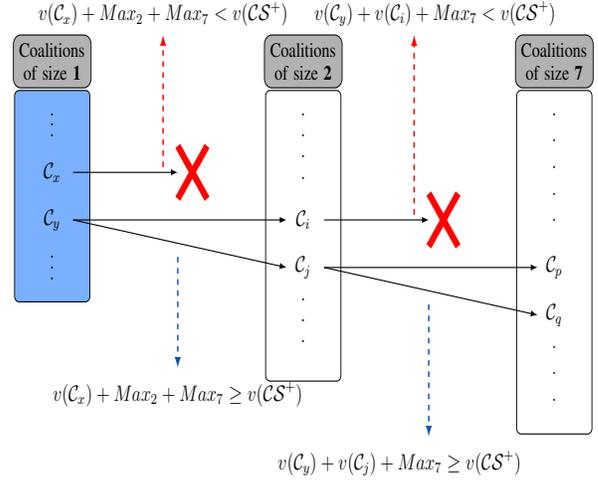

Figure \ref{poi} shows how the POI algorithm (i.e. the modified version of ODP-IP, which is a parallel version that we built), which is used for comparison, works. POI uses the same number of processes as SMART: one process to run the IDP algorithm and the remaining processes to search the integer partition graph.

\begin{figure}[h!]
\small
\centering
\resizebox{0.999\columnwidth}{3.1cm}{%
\begin{tikzpicture}
\large

\node[draw,rectangle,rounded corners=5pt, minimum height=1cm, minimum width=1.7cm, text width=1.6cm, text centered] (a)at(5.0,2.5){POI process 1};

\draw[->,>=latex,red,dashed, thick, line width=1pt] (5.0,1.975) -- (1.5,0.025);

\draw[->,>=latex,linecolor,dashed, thick, line width=1pt] (10.0,1.975) -- (4.0,0.025);

\draw[->,>=latex,vert,dashed, thick, line width=1pt] (15.0,1.975) -- (6.5,0.025);

\draw[->,>=latex,linecolor,dashed, thick, line width=1pt] (10.0,1.975) -- (11.5,0.025);

\draw[->,>=latex,linecolor,dashed, thick, line width=1pt] (10.0,1.975) -- (19,0.025);

\node[draw,rectangle,rounded corners=5pt, minimum height=1cm, minimum width=1.7cm, text width=1.6cm, text centered] (a)at(10.0,2.5){POI process 2};

\draw[->,>=latex,red,dashed, thick, line width=1pt] (5.0,1.975) -- (9.0,0.025);

\draw[->,>=latex,red,dashed, thick, line width=1pt] (5.0,1.975) -- (16.5,0.025);

\node[draw,rectangle,rounded corners=5pt, minimum height=1cm, minimum width=1.7cm, text width=1.6cm, text centered] (a)at(15.0,2.5){POI process 3};

\draw[->,>=latex,vert,dashed, thick, line width=1pt] (15.0,1.975) -- (14,0.025);

\node[draw,rectangle,rounded corners=5pt, minimum height=1cm, minimum width=1.7cm, text width=1.6cm, text centered] (a)at(1.5,-0.5){2,4,4};

\node[draw,rectangle,rounded corners=5pt, minimum height=1cm, minimum width=1.7cm, text width=1.6cm, text centered] (a)at(4.0,-0.5){1,2,2,6};

\node[draw,rectangle,rounded corners=5pt, minimum height=1cm, minimum width=1.7cm, text width=1.6cm, text centered] (a)at(6.5,-0.5){3,3,4};

\node[draw,rectangle,rounded corners=5pt, minimum height=1cm, minimum width=1.7cm, text width=1.6cm, text centered] (a)at(9.0,-0.5){1,1,1,3,4};

\node[draw,rectangle,rounded corners=5pt, minimum height=1cm, minimum width=1.7cm, text width=1.6cm, text centered] (a)at(11.5,-0.5){1,3,6};

\node[draw,rectangle,rounded corners=5pt, minimum height=1cm, minimum width=1.7cm, text width=1.6cm, text centered] (a)at(14.0,-0.5){2,2,3,3};

\node[draw,rectangle,rounded corners=5pt, minimum height=1cm, minimum width=1.7cm, text width=1.6cm, text centered] (a)at(16.5,-0.5){1,2,2,5};
 
\node[draw,rectangle,rounded corners=5pt, minimum height=1cm, minimum width=1.7cm, text width=1.6cm, text centered] (a)at(19.0,-0.5){1,4,5};

\end{tikzpicture}%
}
\caption{Illustration of the subspace distribution technique used in POI (Parallel ODP-IP). The subspaces are organized according to their upper bounds. The subspace [2,4,4] is the highest upper bound node and [1,1,1,1,6] is the lowest upper bound node. The search is distributed between several processes. Each subspace is searched by the first available process. For example in this Figure, three processes share the search.}\label{poi}
\end{figure}
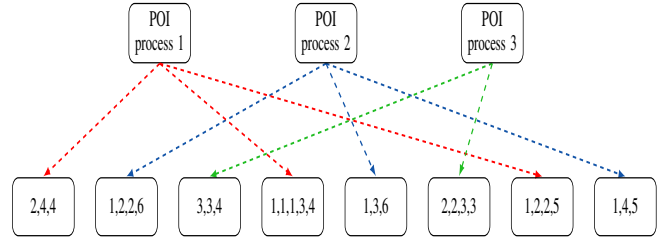

\subsection{Illustration of the subspace pruning in GRAD}

In section 4.2 in the main paper, we explained that after evaluating all the coalitions of size $x$, a GRAD process updates the integer partition graph by adding the edges that result from splitting $x$ into two integers. 
We illustrate this in Figure \ref{IPGSDP_only_hybrid} using a ten-agent integer partition graph. Assume that a GRAD process needs to evaluate the coalitions of sizes $s \in \{2,4,6,10\}$ and assume that, as of now, the GRAD process has only evaluated the coalitions of sizes 2 and 4. 
Given the initial evaluations of SDP, SIP proceeds as follows: it first adds the green edges to the integer partition graph, representing the splitting of the size 2. Once the GRAD process completes the evaluation of the coalitions of size 4, it adds the blue edges resulting from the division of size 4. 
Importantly, these edges are placed wherever it is possible to divide a size 2 or 4 within the graph.
This allows us to connect the green and blue nodes to the bottom node by evaluating the grand coalition. Before evaluating the grand coalition, the nodes are only connected to the gray nodes at level 2.
By evaluating the coalition of size $n$, all the nodes at level 2, namely, $[1,9]$, $[2,8]$, $[3,7]$, $[4,6]$, and $[5,5]$ are now connected to the bottom node through the black edges.
As a result, the gray, green, and blue nodes are now connected to the bottom node.
These nodes are hence fully searched when we evaluate the coalition of size $n$.

\begin{figure*}[h!]
\centering
\resizebox{0.89\textwidth}{14.7cm}{%
\begin{tikzpicture}


\node[draw,rectangle,rounded corners=4pt] (a)at(7.5,0){$1,1,1,1,1,1,1,1,1,1$};

\node[draw,rectangle,rounded corners=4pt] (a)at(7.5,-1.25){$1,1,1,1,1,1,1,1,2$};

\node[draw,rectangle,rounded corners=4pt] (a)at(5.5,-2.5){$1,1,1,1,1,1,1,3$};

\node[draw,rectangle,rounded corners=4pt] (a)at(9.5,-2.5){$1,1,1,1,1,1,2,2$};

\node[draw,rectangle,rounded corners=4pt] (a)at(4.5,-3.75){$1,1,1,1,1,1,4$};

\node[draw,rectangle,rounded corners=4pt] (a)at(7.5,-3.75){$1,1,1,1,1,2,3$};

\node[draw,rectangle,rounded corners=4pt] (a)at(10.5,-3.75){$1,1,1,1,2,2,2$};

\node[draw,rectangle,rounded corners=4pt] (a)at(2,-5){$1,1,1,1,1,5$};

\node[draw,rectangle,rounded corners=4pt] (a)at(4.75,-5){$1,1,1,1,2,4$};

\node[draw,rectangle,rounded corners=4pt] (a)at(7.5,-5){$1,1,1,1,3,3$};

\node[draw,rectangle,rounded corners=4pt] (a)at(10.25,-5){$1,1,1,2,2,3$};

\node[draw,rectangle,rounded corners=4pt] (a)at(13,-5){$1,1,2,2,2,2$};

\node[draw,rectangle,rounded corners=4pt, fill=LightCyan2!95] (a)at(0.75,-6.25){$1,1,1,1,6$};

\node[draw,rectangle,rounded corners=4pt] (a)at(3,-6.25){$1,1,1,2,5$};

\node[draw,rectangle,rounded corners=4pt] (a)at(5.25,-6.25){$1,1,1,3,4$};

\node[draw,rectangle,rounded corners=4pt] (a)at(7.5,-6.25){$1,1,2,2,4$};

\node[draw,rectangle,rounded corners=4pt] (a)at(9.75,-6.25){$1,1,2,3,3$};

\node[draw,rectangle,rounded corners=4pt] (a)at(12,-6.25){$1,2,2,2,3$};

\node[draw,rectangle,rounded corners=4pt] (a)at(14.25,-6.25){$2,2,2,2,2$};

\node[draw,rectangle,rounded corners=4pt] (a)at(-0.5,-7.5){$1,1,1,7$};

\node[draw,rectangle,rounded corners=4pt, fill=LightCyan2!95] (a)at(1.5,-7.5){$1,1,2,6$};

\node[draw,rectangle,rounded corners=4pt] (a)at(3.5,-7.5){$1,1,3,5$};

\node[draw,rectangle,rounded corners=4pt] (a)at(5.5,-7.5){$1,2,2,5$};

\node[draw,rectangle,rounded corners=4pt] (a)at(7.5,-7.5){$1,1,4,4$};

\node[draw,rectangle,rounded corners=4pt] (a)at(9.5,-7.5){$1,2,3,4$};

\node[draw,rectangle,rounded corners=4pt] (a)at(11.5,-7.5){$2,2,2,4$};

\node[draw,rectangle,rounded corners=4pt] (a)at(13.5,-7.5){$1,3,3,3$};

\node[draw,rectangle,rounded corners=4pt] (a)at(15.5,-7.5){$2,2,3,3$};

\node[draw,rectangle,rounded corners=4pt, fill=vert!85] (a)at(0.5,-8.75){$1,1,8$};

\node[draw,rectangle,rounded corners=4pt] (a)at(2.5,-8.75){$1,2,7$};

\node[draw,rectangle,rounded corners=4pt, fill=LightCyan2!95] (a)at(4.5,-8.75){$1,3,6$};

\node[draw,rectangle,rounded corners=4pt, fill=LightCyan2!95] (a)at(6.5,-8.75){$2,2,6$};

\node[draw,rectangle,rounded corners=4pt] (a)at(8.5,-8.75){$1,4,5$};

\node[draw,rectangle,rounded corners=4pt] (a)at(10.5,-8.75){$2,3,5$};

\node[draw,rectangle,rounded corners=4pt] (a)at(12.5,-8.75){$2,4,4$};

\node[draw,rectangle,rounded corners=4pt] (a)at(14.5,-8.75){$3,3,4$};

\node[draw,rectangle,rounded corners=4pt, fill=black!35] (a)at(2.5,-10){$1,9$};

\node[draw,rectangle,rounded corners=4pt, fill=black!35] (a)at(5,-10){$2,8$};

\node[draw,rectangle,rounded corners=4pt, fill=black!35] (a)at(7.5,-10){$3,7$};

\node[draw,rectangle,rounded corners=4pt, fill=black!35] (a)at(10,-10){$4,6$};

\node[draw,rectangle,rounded corners=4pt, fill=black!35] (a)at(12.5,-10){$5,5$};

\node[draw,rectangle,rounded corners=4pt] (a)at(7.5,-11.25){$10$};


\draw[->,>=latex,black] (7.5,-10.95) -- (7.5,-10.3);

\draw[->,>=latex,black] (7.5,-10.95) -- (5,-10.3);

\draw[->,>=latex,black] (7.5,-10.95) -- (2.5,-10.3);

\draw[->,>=latex,black] (7.5,-10.95) -- (10,-10.3);

\draw[->,>=latex,black] (7.5,-10.95) -- (12.5,-10.3);


\draw[->,>=latex,vert!85] (5,-9.7) -- (0.5,-9.05);

\draw[->,>=latex,blue] (10,-9.7) -- (4.5,-9.05);


\draw[->,>=latex,blue] (10,-9.7) -- (6.5,-9.05);


\draw[->,>=latex,vert!85] (2.5,-8.45) -- (-0.5,-7.8);

\draw[->,>=latex,vert!85] (10.5,-8.45) -- (3.5,-7.8);


\draw[->,>=latex,vert!85] (6.5,-8.45) -- (1.5,-7.8);


\draw[->,>=latex,blue] (8.5,-8.45) -- (3.5,-7.8);

\draw[->,>=latex,blue] (8.5,-8.45) -- (5.5,-7.8);

\draw[->,>=latex,vert!85] (12.5,-8.45) -- (7.5,-7.8);
\draw[->,>=latex,blue] (12.5,-8.45) -- (9.5,-7.8);
\draw[->,>=latex,blue] (12.5,-8.45) -- (11.5,-7.8);

\draw[->,>=latex,blue] (14.5,-8.45) -- (13.5,-7.8);
\draw[->,>=latex,blue] (14.5,-8.45) -- (15.5,-7.8);


\draw[->,>=latex,vert!85] (1.5,-7.2) -- (0.75,-6.55);


\draw[->,>=latex,vert!85] (5.5,-7.2) -- (3,-6.55);

\draw[->,>=latex,blue] (7.5,-7.2) -- (5.25,-6.55);

\draw[->,>=latex,blue] (7.5,-7.2) -- (7.5,-6.55);

\draw[->,>=latex,vert!85] (9.5,-7.2) -- (5.25,-6.55);

\draw[->,>=latex,blue] (9.5,-7.2) -- (9.75,-6.55);
\draw[->,>=latex,blue] (9.5,-7.2) -- (12,-6.55);

\draw[->,>=latex,vert!85] (15.5,-7.2) -- (9.75,-6.55);

\draw[->,>=latex,vert!85] (11.5,-7.2) -- (7.5,-6.55);
\draw[->,>=latex,blue] (11.5,-7.2) -- (12,-6.55);
\draw[->,>=latex,blue] (11.5,-7.2) -- (14.25,-6.55);



\draw[->,>=latex,vert!85] (3,-5.95) -- (2,-5.3);

\draw[->,>=latex,blue] (5.25,-5.95) -- (7.5,-5.3);

\draw[->,>=latex,blue] (5.25,-5.95) -- (10.25,-5.3);

\draw[->,>=latex,vert!85] (7.5,-5.95) -- (4.75,-5.3);

\draw[->,>=latex,blue] (7.5,-5.95) -- (10.25,-5.3);

\draw[->,>=latex,blue] (7.5,-5.95) -- (13,-5.3);

\draw[->,>=latex,vert!85] (9.75,-5.95) -- (7.5,-5.3);

\draw[->,>=latex,vert!85] (12,-5.95) -- (10.25,-5.3);

\draw[->,>=latex,vert!85] (14.25,-5.95) -- (13,-5.3);


\draw[->,>=latex,vert!85] (4.75,-4.7) -- (4.5,-4.05);

\draw[->,>=latex,blue] (4.75,-4.7) -- (7.5,-4.05);

\draw[->,>=latex,blue] (4.75,-4.7) -- (10.5,-4.05);

\draw[->,>=latex,vert!85] (10.25,-4.7) -- (7.5,-4.05);

\draw[->,>=latex,vert!85] (13,-4.7) -- (10.5,-4.05);


\draw[->,>=latex,blue] (4.5,-3.45) -- (5.5,-2.8);

\draw[->,>=latex,blue] (4.5,-3.45) -- (9.5,-2.8);

\draw[->,>=latex,vert!85] (7.5,-3.45) -- (5.5,-2.8);

\draw[->,>=latex,vert!85] (10.5,-3.45) -- (9.5,-2.8);


\draw[->,>=latex,vert!85] (9.5,-2.2) -- (7.5,-1.55);

\draw[->,>=latex,vert!85] (7.5,-0.95) -- (7.5,-0.3);

\end{tikzpicture}%
}
\caption{Searched subspaces after evaluation of all coalitions of sizes 2 and 4, given 10 agents. Green and blue colored subspaces are fully searched by SDP, whereas white colored subspaces are not yet searched.} \label{IPGSDP_only_hybrid}
\end{figure*}
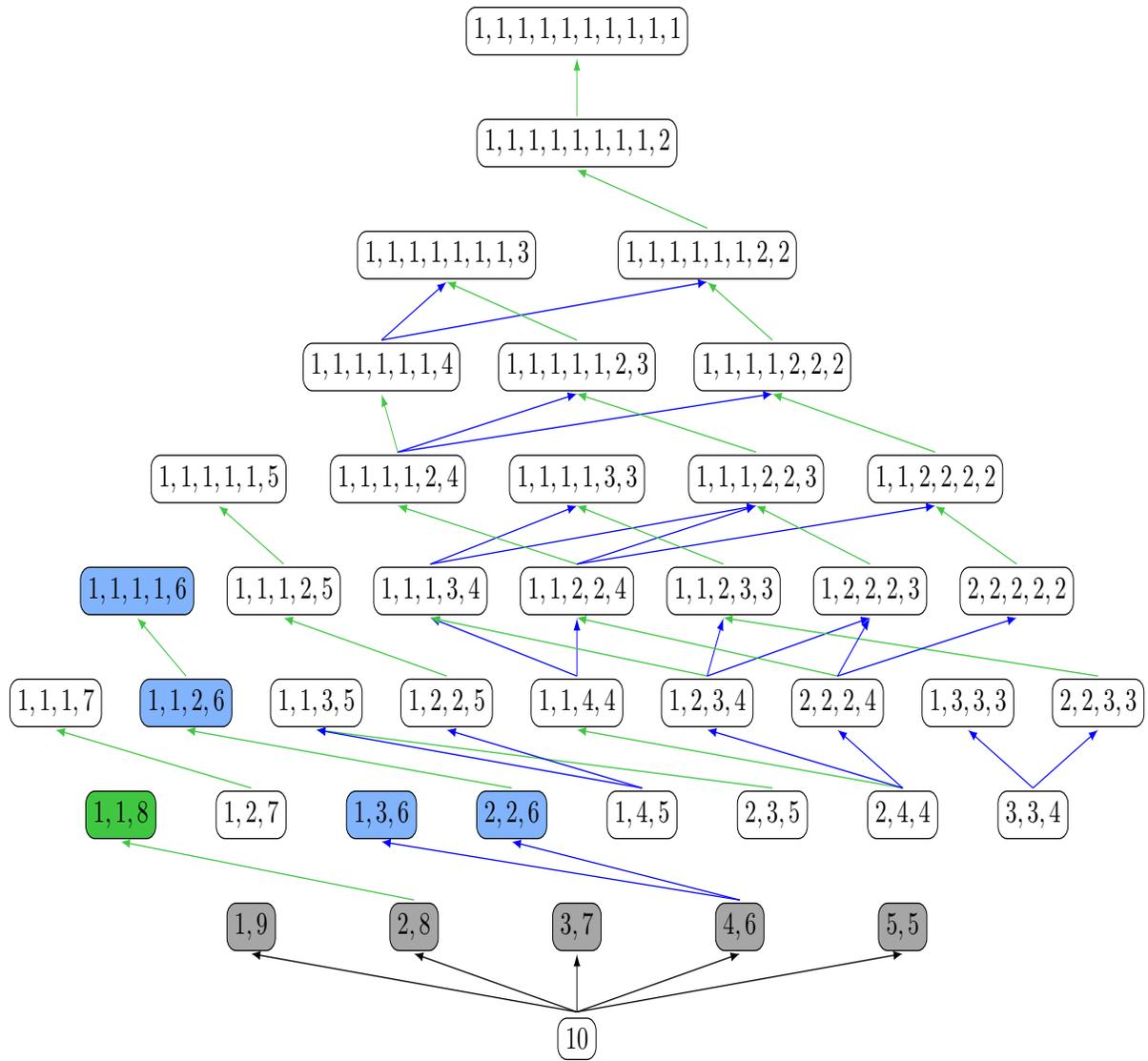

\section{Additional Experimental Results}

As can be seen in Table \ref{tabresultB}, CDP is faster than P-IDP by up to 35\%.

\begin{table*}[h!]
\begin{center}
\renewcommand{\arraystretch}{1.75}
\begin{tabular}{|c|c|c|c|c|c|c|}
\hline
\textbf{Number of Agents} & \multicolumn{4}{c|}{\textbf{Execution Time}} &  \multicolumn{2}{c|}{\textbf{Time Gain}} \\
\cline{2-5}
 &  \textbf{CDP} & \textbf{IDP} & \textbf{ODP} & \textbf{P-IDP} & \textbf{(IDP)} & \textbf{(P-IDP)} \\
\hline
20 & $1.7$ & $3.7$ & $2.2$ & $2.0$ & 54\% & $15\%$ \\
\hline
21 & $7.4$ & $15.9$ & $10.3$ & $7.9$ & 52\% & $7\%$ \\
\hline
22 & $12.3$ & $24.7$ & $19.8$ & $15.7$ & 50\% & $22\%$ \\
\hline
23 & $57$ & $131$ & $79$ & $69$ & 57\% & $18\%$ \\
\hline
24 & $205$ & $507$ & $382$ & $257$ & 60\% & $21\%$ \\
\hline
25 & $427$ & $887$ & $675$ & $593$ & 52\% & $28\%$ \\
\hline
26 & $1781$ & $3659$ & $2846$ & $2267$ & 52\% & $22\%$ \\
\hline
27 & $3390$ & $7078$ & $5508$ & $5196$ & 53\% & $35\%$ \\
\hline
\end{tabular}
\end{center}
\caption{Time in seconds of CDP, IDP, ODP and P-IDP. Time  Gain column shows the time gain achieved by CDP compared to IDP and P-IDP.}\label{tabresultB}
\end{table*}

Figure \ref{additionalResults} shows additional results of SMART on other value distributions: Zipf, SVA-Beta~\cite{changder2020odss}.

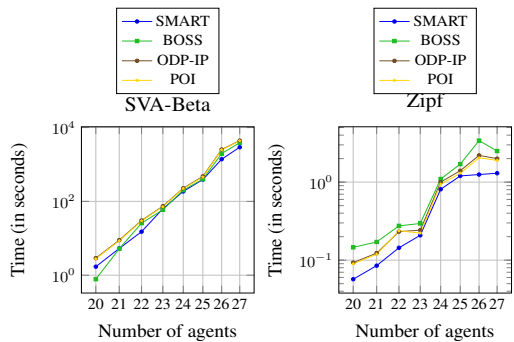
\begin{figure}[h!]
\begin{center}
\small
     \resizebox{0.39\columnwidth}{4.5cm}{%
     \begin{subfigure}{0.28\textwidth}
         \centering
         \begin{tikzpicture}
	\begin{loglogaxis}[	grid= major ,
            title=SVA-Beta ,
			width=0.8\textwidth ,
            xlabel = {Number of agents} ,
			ylabel = {Time (in seconds)} ,
            width=5cm,height=5cm,
            xtick={20,21,22,23,24,25,26,27},
            xticklabels={20,21,22,23,24,25,26,27},
            ymode=log,
            label style={font=\large},
			title style={font=\Large},
			tick label style={font=\footnotesize},
            legend entries={SMART, BOSS, ODP-IP, POI},
			legend style={at={(0.5,1.715)},anchor=north}]
			\addplot+[ultra thin,mark size=1pt] coordinates {(20,1.7) (21,5.3) (22,15.1) (23,63) (24,186) (25,389) (26,1387) (27,2932) }; 
			\addplot+[ultra thin,mark size=1pt,vert,mark options={fill=vert}] coordinates {(20,0.78) (21,5.2) (22,25.7) (23,59) (24,201) (25,398) (26,1973) (27,3872) };
			\addplot+[ultra thin,mark size=1pt] coordinates {(20,2.9) (21,8.9) (22,30.4) (23,74) (24,228) (25,476) (26,2546) (27,4429) };
			\addplot+[ultra thin,mark size=1pt,gold] coordinates {(20,2.8) (21,8.6) (22,29.6) (23,72) (24,224) (25,454) (26,2499) (27,4376) };
			
	\end{loglogaxis}
    \end{tikzpicture}\\
         \label{fig:three sin x}
     \end{subfigure}
     }
     \resizebox{0.39\columnwidth}{4.5cm}{%
     \begin{subfigure}{0.28\textwidth}
         \centering
         \begin{tikzpicture}
	\begin{loglogaxis}[	grid= major ,
            title=Zipf,
			width=0.6\textwidth ,
            xlabel = {Number of agents} ,
			ylabel = {Time (in seconds)} ,
            width=5cm,height=5cm,
            xtick={20,21,22,23,24,25,26,27},
            xticklabels={20,21,22,23,24,25,26,27},
            ymode=log,
            label style={font=\large},
			title style={font=\Large},
			tick label style={font=\footnotesize},
            legend entries={SMART, BOSS, ODP-IP, POI},
			legend style={at={(0.5,1.715)},anchor=north}]
			\addplot+[ultra thin,mark size=1pt] coordinates {(20,0.057) (21,0.085) (22,0.144) (23,0.208) (24,0.81) (25,1.2) (26,1.25) (27,1.3) };
			\addplot+[ultra thin,mark size=1pt,vert,mark options={fill=vert}] coordinates {(20,0.146) (21,0.171) (22,0.275) (23,0.295) (24,1.096) (25,1.7) (26,3.4) (27,2.5) };
			\addplot+[ultra thin,mark size=1pt] coordinates {(20,0.093) (21,0.123) (22,0.232) (23,0.242) (24,1) (25,1.4) (26,2.2) (27,2) };
			\addplot+[ultra thin,mark size=1pt,gold] coordinates {(20,0.090) (21,0.119) (22,0.238) (23,0.222) (24,0.94) (25,1.3) (26,2.05) (27,1.9) };
	\end{loglogaxis}
    \end{tikzpicture}\\
         \label{fig:y equals x}
     \end{subfigure}
     }

             \caption{Time performance in seconds of SMART, BOSS, ODP-IP and POI for a number of agents between 20 and 27.}
        \label{additionalResults}
\normalsize
\end{center}
\end{figure}


\end{document}